\documentclass[11pt]{article}
\usepackage[]{amsmath,amssymb,amsfonts,latexsym,amsthm,enumerate,fullpage}

\usepackage{color}

\newtheorem{prelem}{{\bf Theorem}}

\newtheorem{theorem}{Theorem}[section]

\newtheorem{definition}[theorem]{Definition}

\newtheorem{claim}[theorem]{Claim}

\newtheorem{lemma}[theorem]{Lemma}
\newtheorem{proposition}[theorem]{Proposition}
\newtheorem{observation}[theorem]{Observation}
\newtheorem{prob}[theorem]{Problem}

\newtheorem{remarka}[theorem]{Remark}
\newenvironment{remark}{\begin{remarka}\rm}{\hfill\rule{2mm}{2mm}\end{remarka}}
\newtheorem{examplea}[theorem]{Example}

\newtheorem*{thm}{Informal Statement}

\def\Aff {{\rm Aff}}
\def\Lin {{\rm Lin}}

\def\span {{\rm span}}

\def\Ex {{\mathbb E}}
\def\Pr {{\rm Pr}}

\def\bias{\textrm{bias}}

\def\rank{{\rm rank}}
\def\bias{{\rm bias}}

\def\exp{{\mathrm{e}_p}}

\def\e{{e}} 
\def\x{\mathbf{x}}
\def\X{\mathbf{X}}
\def\Y{\mathbf{Y}}

\def\P{\mathcal{P}}
\def\Q{\mathcal{Q}}

\def\F{\mathbb{F}}
\def\C{\mathbb{C}}
\def\eps{\varepsilon}
\def\N{\mathbb{N}}

\def\D{\mathbb{D}}
\def\E{\mathbb{E}}

\def\ff{\mathrm{f}}
\renewcommand{\gg}{\mathrm{g}}
\def\hh{\mathrm{h}}

\def\FF{\mathrm{F}}

\def\Dp{P(\F_p)}

\def\conj{\mathcal{C}}

\newcommand{\dotcup}{\ensuremath{\mathaccent\cdot\cup}}

\newcommand{\ip}[1]{\langle#1\rangle}
\newcommand{\ignore}[1]{}

\title{Correlation Testing for Affine Invariant Properties on $\mathbb{F}_p^n$ in the High Error Regime
\footnote{A preliminary version of this work appeared in STOC' 2011.}
}
\author{Hamed Hatami\thanks{Supported by an NSERC and an FQRNT grant.}\\
School of Computer Science, McGill University, Montr\'eal,  Canada\\
hatami@cs.mcgill.ca
\and
Shachar Lovett\thanks{Supported by NSF grant DMS-0835373.}\\
School of Mathematics, Institute of Advanced Study, Princeton, USA\\
slovett@math.ias.edu}

\date{}
\begin{document}
\maketitle

\begin{abstract}

Recently there has been much interest in Gowers uniformity norms from the perspective of theoretical computer science. This is mainly due to the fact that these norms provide a method for testing whether the maximum correlation of a function $f:\F_p^n \rightarrow \F_p$ with polynomials of degree at most $d \le p$ is non-negligible, while making only a constant number of queries to the function. This is an instance of {\em correlation testing}. In this framework, a fixed test is applied to a function, and the acceptance probability of the test is dependent on the correlation of the function from the property. This is an analog of {\em proximity oblivious testing}, a notion coined by Goldreich and Ron, in the high error regime.

In this work, we study  general properties which are affine invariant and which are correlation testable using a constant number of queries. We show that any such property (as long as the field size is not too small) can in fact be tested by  Gowers uniformity tests, and hence having correlation with the property is equivalent to having correlation with degree $d$ polynomials for some fixed $d$. We stress that our result holds also for non-linear properties which are affine invariant. This completely classifies affine invariant properties which are correlation testable.

The proof is based on higher-order Fourier analysis. Another ingredient is a nontrivial extension of a graph theoretical theorem of Erd\"os, Lov\'asz and Spencer to the context of additive number theory.
\end{abstract}

\newline
{{\sc Keywords:} property testing, higher-order Fourier analysis;


\newpage
\tableofcontents
\newpage


\section{Introduction \label{sec:intro}}

Blum, Luby, and Rubinfeld~\cite{BLR} made a beautiful observation that given a function $f:\F_p^n \rightarrow \F_p$, it is possible to inquire the value of $f$ on a few random points, and accordingly probabilistically distinguish between the case that $f$ is a linear function and the case that $f$ has to be modified on at least $\eps>0$ fraction of points to become a linear function. Inspired by this observation, Rubinfeld and Sudan~\cite{866666} defined the concept of property testing which is now a major area of research in theoretical computer science. Roughly speaking to test a function for a property means to examine the value of the function on a few random points, and accordingly (probabilistically) distinguish between
the case that the function has the property and the case that it is not too close to any function with that property.  Interestingly and to some extent surprisingly these tests exist for various basic properties. The first substantial investigation of property testing occurred in Goldreich, Goldwasser, and Ron~\cite{285060} who showed that several  natural combinatorial properties are testable. Since then there has been a significant amount of research on classifying the testable properties in combinatorial and algebraic settings.

The theory of uniformity plays an important role in the area of property testing.
In~\cite{MR1844079} Gowers in a Fourier-analytic proof for Szemer\'edi's theorem introduced a new notion of uniformity, defined through Gowers uniformity norms. This notion has an interesting implication in the context of property testing. The Gowers norm of a function can be expressed as an average of values of this function on a few random sample points. There are  theorems which show that a bounded function has non-negligible Gowers uniformity norm if and only if it has a non-negligible correlation with a low degree polynomial. These two facts show that having correlation with low degree polynomials is testable. Our main result is that, roughly speaking, the only correlation testable families of functions are the ones that can be tested by Gowers uniformity norms.

The tests that we study in this article are slightly different in nature from the typical statements in the area of property testing. Typically in property testing the goal is to distinguish the functions that are in a set $\mathcal{D}$ (structured) from the functions that are in a non-negligible distance  from every element in $\mathcal{D}$. In this article, we are interested in a different kind of tests. Here we seek a weaker structure in the function, and having correlation with an element in $\mathcal{D}$ replaces the usual condition of actually being in $\mathcal{D}$. This is closely related to the   concept of tolerant property testing~\cite{MR2252941}. Also our tests use a fixed number of queries and in this sense are in the spirit of proximity oblivious testing~\cite{GoldreichRon-2009}.

\subsection{Notations}
For a natural number $k$, denote $[k]:=\{1,\ldots,k\}$. The complex unit disk is denoted by $\mathbb{D}=\{z \in \C: |z| \le 1\}$. We will usually use the lower English letters $x,y,z$ to denote the elements of $\F_p^n$. For $x \in \F_p^n$, and $i \in [n]$, $x(i)$ denotes the $i$-th coordinate of $x$, i.e. $x=(x(1),\ldots,x(n))$. We frequently need to work with the elements of $(\F_p^n)^k$ which we regard as vectors with $k$ coordinates. These elements are denoted with bold font e.g. $\x=(x_1,\ldots,x_k) \in (\F_p^n)^k$. For $1 \le i \le n$, $\e_i$ denotes the $i$-th standard vector in $\F_p^n$. Capital letters $X$, $Y$, \emph{etc} are used to denote random variables.

For an element $m \in \F_p$, we use the notation $\exp(m):=e^{\frac{2 \pi i}{p}m}$. We denote by $f, g,\ldots$ functions from $\F_p^n$ to $\C$, and by $\ff, \gg,\ldots$ functions from $\F_p^n$ to $\F_p$. We denote $f:=\exp(\ff)$ if $f(x)=\exp(\ff(x))$. We denote $n$-variate polynomials over $\F_p^n$ by $P,Q,\ldots$. 

The \emph{bias} of a function $f:\F_p^n \rightarrow \C$ is defined to be the quantity
\begin{equation}
\label{eq:biasDef}
\bias(f) := \left| \Ex_{X \in \F_p^n}[f(X)] \right|.
\end{equation}
The \emph{bias} of a function $\ff:\F_p^n \rightarrow \F_p$ is defined to be $\bias(\ff):=\bias(\exp(\ff))$.
The inner product of two functions $f,g:\F_p^n \rightarrow \C$ is defined as
\begin{equation}
\label{eq:corrDef}
\ip{f,g} := \Ex_{X \in \F_p^n}[f(X)\overline{g(X)}].
\end{equation}
The \emph{correlation} of a function $f:\F_p^n \rightarrow \C$ with a set $D$ of functions from $\F_p^n$ to $\C$ is defined as
\begin{equation}
\label{eq:corrDefSet}
\|f\|_{u(D)} :=\sup_{g \in D} \left| \ip{ f,g } \right|.
\end{equation}
By an abuse of notation, if $D$ is a set of functions from $\F_p^n$ to $\F_p$, we define
$$
\|f\|_{u(D)}:=\sup_{\gg \in D} \left| \ip{ f,\exp(\gg) } \right|.
$$
Note that $\|\cdot\|_{u(D)}$ is always a semi-norm.

For two vector spaces $V$ and $W$ over $\F_p$, let $\Lin(V,W)$ denote the set of all \emph{linear transformations} from $V$ to $W$. Let $\Aff(n, \F_p)$ denote the group of all \emph{invertible affine transformations} from $\F_p^n$ to itself. For a function $f:\F_p^n \to \C$, and an $A \in \Aff(n,\mathbb{F}_p)$, we denote by $Af$ the function that maps $x$ to $f(Ax)$.\footnote{Here we think of $A$ as the operator that maps the function $f$ to the function $Af$. This explains our choice of notation in using $Af$ rather than $fA$.}

\subsection{Gowers Uniformity Norms}
Gowers uniformity norms are defined in the more general setting of arbitrary finite Abelian groups, but in this article we are only interested in the case where the group is $\F_p^n$.
\begin{definition}[Gowers uniformity norms]
Let $G$ be a finite Abelian group and $f:G \rightarrow \C$.  For an integer $k \ge 1$, the $k$-th Gowers norm of $f$, denoted $\|f\|_{U^k}$ is defined by
\begin{equation}
\label{eq:GowersNorm}
\|f\|_{U^k}^{2^k} := \Ex \left[ \prod_{S \subseteq [k]} \mathcal{C}^{k-|S|} f\left(X + \sum_{i \in S} Y_i\right) \right],
\end{equation}
where $\mathcal{C}$ denotes the complex conjugation operator, and $X,Y_1,\ldots,Y_k$ are independent random variables taking values in $G$ uniformly at random.
\end{definition}
These norms were first defined in~\cite{MR1844079} in the case where $G$ is the group $\mathbb{Z}_N$. Note that $\|f\|_{U^1}=\left|\Ex[f(X)]\right|$, and thus $\|\cdot\|_{U^1}$
is a semi-norm rather than a norm.
The facts that the right-hand side of (\ref{eq:GowersNorm}) is always non-negative, and that for $k>1$, $\| \cdot \|_{U^k}$ is actually a norm are easy to prove, but not trivial (see~\cite{MR1844079} for a proof).

Let us explain the relevance of the uniformity norms to the area of property testing. The goal in property testing is to obtain certain information about a function by ``reading'' its values only on a small number of points. The simplest case is that of testing correlation with linear functions. For $a \in \F_p^n$, let $\ell_a:\F_p^n \to \F_p$ be the corresponding linear function, defined as $\ell_a(x)=\sum_{i=1}^n a(i) x(i)$. Let $\mathrm{Linear}=\{\ell_a: a \in \F_p^n\}$ be the set of linear functions. Let $\ff:\F_p^n \to \F_p$ be a function, and let $f=\exp(\ff)$. The correlation of $\ff$ with linear functions is given by
$$
\|f\|_{u(\rm{Linear})} = \max_{\ell_a \in \mathrm{Linear}} |\ip{f,\exp(\ell_a)}| =
\max_{a \in \F_p^n} \bias(\ff-\ell_a),
$$
which is the same as the maximum of the absolute values of the Fourier coefficient of $f$. It is known~\cite{BLR,MR1844079} that the correlation of $\ff$ with linear functions is related to the $U^2$ norm of $f$. Specifically, for every $\eps>0$,
\begin{itemize}
\item \emph{Direct Theorem:} If $\|f\|_{u(\rm{Linear})} \ge \eps$, then $\|f\|_{U^2} \ge \eps.$
\item \emph{Inverse Theorem:} If $\|f\|_{U^2} \ge \eps$, then $\|f\|_{u(\rm{Linear})} \ge \eps^2.$
\end{itemize}
These two facts together show that $\|f\|_{U^2}$ gives a rough estimate for the maximum correlation of $\ff$ with linear functions. Recall that
$$
\|f\|_{U^2}^4 =
\|\exp(\ff)\|_{U^2}^4=\Ex[\exp(\ff(X+Y+Z)-\ff(X+Y)-\ff(X+Z)+\ff(X))],
$$
where $X,Y,Z \in \F_p^n$ are uniformly chosen. In other words, the joint distribution of $\ff(X+Y+Z), \ff(X+Y), \ff(X+Z), \ff(X)$ allows to distinguish between the case that $\ff$ has correlation at least
$\eps$ with linear functions, and the case where $\ff$ has correlation at most $\delta=\eps^2/2$ (say) with linear functions. Hence, correlation with linear functions is ``testable with just $4$ queries to $\ff$'' (for every $\eps>0$).

Similar to the case of linear functions, correlation with degree $d$ polynomials can be tested by the $U^{d+1}$ uniformity norm. The main ingredient is the inverse theorem for $\F_p^n$, recently proved in~\cite{bergelson-2009, tao-2008}.
Let $\textrm{Poly}_d(\F_p^n)$ be the set of polynomials of degree at most $d$ over $\F_p^n$. These results show that, as long as $p>d$ (i.e. the field is not too small~\footnote{Recently, the inverse theorem for the case $p \le d$ was proved in~\cite{TaoZiegler2010}. In this range the role of polynomials is replaced by {\em non-classical polynomials}. As our techniques are limited to the case of $p>d$, we do not expand on this.
}),  for every $\eps>0$,
\begin{itemize}
\item \emph{Direct Theorem:} If $\|f\|_{u(\textrm{Poly}_d(\F_p^n))} \ge \eps$, then $\|f\|_{U^{d+1}} \ge \eps.$
\item \emph{Inverse Theorem:} If $\|f\|_{U^{d+1}} \ge \eps$, then $\|f\|_{u(\textrm{Poly}_d(\F_p^n))} \ge \delta(\eps)>0.$
\end{itemize}

The exact dependency of $\delta(\eps)$ on $\eps$ is currently unknown for $d \ge 3$, but crucially $\delta$ does not depend on $n$. Analogously to the case of linear functions and $U^2$, we get that the joint distribution of $(\ff(X+\sum_{i \in I} Y_i): I \subseteq [d+1])$ where $X,Y_1,\ldots,Y_{d+1} \in \F_p^n$ are uniformly chosen, distinguishes the case where $\ff$ has noticeable correlation ($\ge \eps$) with degree $d$ polynomials, from the case where $\ff$ has negligible correlation ($\le \delta(\eps)/2$) with degree $d$ polynomials, for every $\eps>0$. Hence, correlation with polynomials of total degree $d$ is testable with just $2^{d+1}$ queries to $\ff$ (for every $\eps>0$).

An equivalent qualitative formulation of the direct and inverse theorems stated above is as follows. Let $(\ff_n:\F_p^n \to \F_p)_{n \in \N}$ be a sequence of functions. Then as long as $d<p$ we have that
$$
\lim_{n \to \infty} \|\exp(\ff_n)\|_{u(\textrm{Poly}_d(\F_p^n))}=0 \Longleftrightarrow \lim_{n \to \infty} \|\exp(\ff_n)\|_{U^{d+1}}=0.
$$
It will be convenient to express our results in such terms.

\subsection{Correlation testable properties \label{sec:introCorTesting}}
\label{subsec:intro_correlation_testable_properties}
The above discussion motivates us to introduce a general definition of correlation testable properties. Let $\mathcal{D}=\{D_n\}_{n \in \N}$ be a family of sets, where each $D_n$ is a set of functions from $\F_p^n$ to $\F_p$. Informally, $\mathcal{D}$ is correlation testable using $q$ queries if there exists a distribution over $x_1,\ldots,x_q \in \F_p^n$, and for every $\eps>0$ there exists $\delta(\eps) \in (0,\eps)$, such that the following holds. The joint distribution of $(\ff(x_1),\ldots,\ff(x_q))$ allows to distinguish between the case that $\ff$ has noticeable correlation ($\ge \eps$) with $\mathcal{D}$ and the case that $\ff$ has negligible correlation ($\le \delta(\eps)$) with $\mathcal{D}$.

\begin{definition}[Correlation testable properties]
\label{def:corTestable}
A family $\mathcal{D}=(D_n)$ is correlation testable with $q$ queries, if there exists a distribution $\mu$ taking values in $(\F_p^n)^q$ and a mapping $\Gamma:\F_p^q \to \{0,1\}$, such that the following holds. For every $\eps>0$, there exist $\delta \in (0,\eps)$, $0 \le \theta^- < \theta^+ \le 1$ and $n_0 \in \N$, such that for every $n>n_0$ and $\ff:\F_p^n \to \F_p$ we have:
\begin{itemize}
\item If $\|\exp(\ff)\|_{u(D_n)} \ge \eps$ then $\Pr_{(X_1,\ldots,X_q) \sim \mu}[\Gamma(\ff(X_1),\ldots,\ff(X_q))=1] \ge \theta^+$.
\item If $\|\exp(\ff)\|_{u(D_n)} \le \delta$ then $\Pr_{(X_1,\ldots,X_q) \sim \mu}[\Gamma(\ff(X_1),\ldots,\ff(X_q))=1] \le \theta^-$.
\end{itemize}
\end{definition}

Following the discussion on Gowers uniformity norms, $\mathrm{Poly}_d = \{\mathrm{Poly}_d(\F_p^n)\}_{n \in \N}$ is correlation testable using $q=2^{d+1}$ queries, as long as $d<p$. More precisely, let $X,Y_1,\ldots,Y_d$ be random variables taking values in $\F_p^n$ uniformly at random. Then defining $\mu$ to be the distribution of $(X_1,\ldots,X_{2^{d+1}}):=(X+\sum_{i \in S} Y_i)_{S \subseteq [d+1]}$, and defining $\Gamma$  as $\Gamma((f(X+\sum_{i \in S} Y_i))_{S \subseteq [d+1]})=1$ if and only if $\sum_{S \subseteq [d+1]} (-1)^{|S|}f(X+\sum_{i \in S} Y_i)= 0$ provides a correlation test for $\mathrm{Poly}_d$. 

Our goal is to study families $\mathcal{D}=\{D_n\}_{n \in \N}$ which are correlation testable using a constant number of queries. We shall require the sets $D_n$ to be \emph{consistent} with each other:

\begin{itemize}
\item {\bf A1: (Consistency)} For positive integers $m>n$ and $\gg \in D_n$, the function $\hh: \F_p^m \rightarrow \F_p$ defined as $\hh(x_1,\ldots,x_m)=\gg(x_1,\ldots,x_n)$ belongs to $D_m$.
\end{itemize}
We need to make also a more crucial  assumption. In this general setting, the algebraic structure of $\F_p^n$ is ignored, and we are treating  $\F_p^n$ as a generic set of size $p^n$. In order to take  the algebraic structure of $\F_p^n$ into account, we shall require $D_n$ to be \emph{affine invariant}:
\begin{itemize}
\item {\bf A2: (Affine invariance)} For every positive integer $n$, if $\gg \in D_n$, then for every $A \in \Aff(n,\F_p)$, we have $A\gg \in D_n$.
\end{itemize}
In our context, the set $\mathrm{Poly}_d$ is probably the most important example of an affine invariant property.
Invariance plays a crucial role in the area of property testing. We refer to the work of Kaufman and Sudan~\cite{KaufmanSudan-invariance} for the role of invariance in algebraic property testing. We stress that we {\bf do not} require $\mathcal{D}$ to be linear; i.e. we do not require that if $\ff,\gg \in D_n$ then also $\ff+\gg \in D_n$.

The last condition relates to the size of $D_n$. We study families where one can distinguish functions with noticeable correlation from functions with negligible correlation. In order not to make this meaningless, we would like to have functions with negligible correlation. For example, we would like a random function not to have correlation with $\mathcal{D}$ with high probability. Fix a function $\gg \in D_n$. The number of functions $\ff:\F_p^n \to \F_p$ which have correlation at least $\delta$ with $\gg$ is $p^{c(\delta) \cdot p^n}$ where $\lim_{\delta \to 0} c(\delta)=1$. Thus, a sufficient condition for random functions not to be correlated with $D_n$ is that the size of $D_n$ is $p^{o(p^n)}$. We thus add the following requirement.
\begin{itemize}
\item {\bf A3: (Sparsity)} For every $\eps>0$ and large enough $n$, we have $|D_n| \le p^{\eps \cdot p^n}$.
\end{itemize}
We call every $\mathcal{D} = \{D_n\}_{n \in \mathbb{N}}$ satisfying assumptions $\bf A1, \bf A2, \bf A3$ a \emph{proper dual}.

Our main result is the following theorem which roughly speaking says that the only correlation testable families of  functions are the ones that can be tested by Gowers uniformity norms provided that the field size $p$ is not too small. Say a sequence of functions $(\ff_n)_{n \in \N}$ is {\em unbiased} if $\lim_{n \to \infty} \bias(\ff_n)=0$.

\begin{theorem}[Main theorem]
\label{thm:F:main}
Consider a proper dual $\mathcal{D}=\{D_n\}_{n \in \N}$. If $\mathcal{D}$ is correlation testable with $q$ queries,
and if $p \ge q-1$, then there exists $1 \le k \le q-1$ such that the following holds. For every unbiased sequence of functions $(\ff_n:\F_p^n \rightarrow \F_p)_{n \in \N}$, we have
$$
\lim_{n \to \infty} \|\exp(\ff_n)\|_{u(D_n)}=0 \Longleftrightarrow\lim_{n \to \infty} \|\exp(\ff_n)\|_{U^{k+1}}=0.
$$
\end{theorem}
\noindent
Equivalently, Theorem~\ref{thm:F:main} can be phrased as follows: the sequence $(\ff_n)_{n \in \N}$ has correlation
with $\mathcal{D}$ if and only if it has correlation with polynomials of degree at most $k$.

We remark on the seemingly odd requirement that $\ff_n$ is unbiased. Let $\mu$ be some distribution over $\F_p$, and let $\ff_n:\F_p^n \to \F_p$ be a random function, where each $\ff_n(x)$ is sampled independently according to $\mu$. By condition $\bf A3$, since $|D_n|=p^{o(p^n)}$, we have that with probability $1-o_n(1)$,
$$
\lim_{n \to \infty} \|\exp(\ff_n)\|_{u(D_n)}=0.
$$
However, if $\mu$ has a nonzero bias, then almost surely $\ff_n$ will have correlation with constant functions. It turns out
that ruling out sequences of functions which have correlation with constant functions is sufficient for establishing Theorem~\ref{thm:F:main}.

\paragraph{Paper organization}
We give a short overview of the proof of our results in Section~\ref{sec:proof_overview}.
We discuss systems of linear forms in Section~\ref{sec:systemLinearForms}. We survey higher-order Fourier analysis in Section~\ref{sec:highFourier}. We prove the main result for correlation testable families in Section~\ref{sec:F:testing_correlation}. We prove the extension of Erd\"os-Lov\'asz-Spencer theorem in Section~\ref{sec:interior}. We give some concluding remarks and pose some open problems in Section~\ref{sec:conclusion}.

\section{Proof overview}
\label{sec:proof_overview}
In this section, we give an overview of the proofs of our main results. We skip most of the technicalities, and try to emphasis the basic ideas behind the proofs. If the reader wishes, he/she can skip this section to go directly to Section~\ref{sec:systemLinearForms} to read the formal definitions and proofs.

The proof of Theorem~\ref{thm:F:main} is based on studying averages of functions evaluated on linear forms.

\subsection{Linear forms \label{sec:linforms}}
A \emph{linear form} in $k$ variables is a vector $L=(\lambda_1,\ldots,\lambda_k) \in \F_p^k$ regarded as a linear function from $V^k$ to $V$, for every vector space $V$ over $\F_p$: If $\x=(x_1,\ldots,x_k) \in V^k$, then $L(\x) := \lambda_1 x_1+\ldots+\lambda_k x_k$. A \emph{system of linear forms}  in $k$ variables is a finite set $\mathcal{L}=\{L_1,\ldots,L_m\}$ of distinct linear forms $L_i$  in $k$ variables.  For a function $f:\F_p^n \rightarrow \C$, and a system of linear forms $\mathcal{L} = \{L_1,\ldots,L_m\}$ in $k$ variables, define the average
\begin{equation}
\label{eq:intro:avgLinearForms}
t_\mathcal{L}(f) := \Ex\left[  \prod_{i=1}^m f(L_i(\X)) \right],
\end{equation}
where $\X$ is a random variable taking values uniformly in $(\F_p^n)^k$. We define two generalizations of such averages. First, in the case of complex-valued functions we allow taking conjugations. For $\alpha \in \{0,1\}^m$, define
\begin{equation}
\label{eq:intro:avgLinearForms1}
t_\mathcal{L,\alpha}(f) := \Ex\left[  \prod_{i=1}^m \mathcal{C}^{\alpha(i)} f(L_i(\X)) \right],
\end{equation}
where $\mathcal{C}$ is the conjugation operator. Second, if $\ff:\F_p^n \to \F_p$, one may take coefficients in $\F_p$. For $\beta \in \F_p^n$, define
\begin{equation}
\label{eq:intro:avgLinearForms2}
t^*_\mathcal{L,\beta}(\ff) := \Ex\left[  \exp\left(\sum_{i=1}^m  \beta(i) \ff(L_i(\X)) \right)\right].
\end{equation}
We note that for functions $\ff:\F_p^n \to \F_p$,~\eqref{eq:intro:avgLinearForms2} is more general than~\eqref{eq:intro:avgLinearForms1}. Indeed, if $\alpha \in \{0,1\}^m$, then setting $\beta(i)=(-1)^{\alpha(i)}$ gives
$t_{\mathcal{L},\alpha}(\exp(\ff))=t^*_{\mathcal{L},\beta}(\ff)$.

Let $\mathcal{D}$ be a proper dual family which is correlation testable. Then since $\mathcal{D}$ is affine invariant, it is easy to show that correlation with $\mathcal{D}$ can be expressed in terms of averages $t^*_{\mathcal{L},\beta}$. We show in Lemma~\ref{lemma:F:correlation_testable_by_linear_forms_with_few_queries} that
for every $\eps>0$ there exist $\delta \in (0,\eps)$, systems of linear forms $\mathcal{L}_1,\ldots,\mathcal{L}_{\ell}$ and corresponding coefficients $\beta_1,\ldots,\beta_{\ell}$, and $n_0 \in \N$, such that the closures of the following two sets are disjoint:
$$T_\eps^* := \{(t^*_{\mathcal{L}_1,\beta_1}(\ff),\ldots, t^*_{\mathcal{L}_{\ell},\beta_{\ell}}(\ff)) | n >n_0,
\ff:\F_p^n \rightarrow \F_p, \|\exp(\ff)\|_{u(D_n)} \ge \eps\},$$
and
$$S_\eps^* := \{(t^*_{\mathcal{L}_1,\beta_1}(\ff),\ldots, t^*_{\mathcal{L}_{\ell},\beta_{\ell}}(\ff)) | n>n_0, \ff:\F_p^n \rightarrow \F_p, \|\exp(\ff)\|_{u(D_n)} \le \delta\}.$$
Moreover, each of the systems $\mathcal{L}_1,\ldots,\mathcal{L}_{\ell}$ has at most $q$ linear forms. Note that we use the notation $S^*_\eps$ instead of $S^*_\delta$. This is a valid notation as $\delta$ is a function of $\eps$.

Recall that $\D$ is the unit disk in $\C$. It turns out to be easier to first analyze sets with a stronger
requirement which holds for all functions $f:\F_p^n \to \D$, and not just for functions of the form $f=\exp(\ff)$ for $\ff:\F_p^n \to \F_p$. Say a family $\mathcal{D}$ is {\em strongly correlation testable} if the closures of the following two sets are disjoint:
$$T_\eps := \{(t_{\mathcal{L}_1,\alpha_1}(f),\ldots, t_{\mathcal{L}_{\ell},\alpha_{\ell}}(f)) | n >n_0,
f:\F_p^n \rightarrow \D, \|f\|_{u(D_n)} \ge \eps\},$$
and
$$S_\eps := \{(t_{\mathcal{L}_1,\alpha_1}(f),\ldots, t_{\mathcal{L}_{\ell},\alpha_{\ell}}(f)) | n >n_0,
f:\F_p^n \rightarrow \D, \|f\|_{u(D_n)} \le \delta\}.$$

\subsection{Approximation by averages over polynomial factors}
We turn to analyze averages of the form $t_{\mathcal{L},\alpha}(f)$, and more generally averages of the form $\Ex \left[\prod_{i=1}^m \mathcal{C}^{\alpha(i)} f_i(L_i(\X)) \right]$ where $\X \in (\F_p^n)^k$ is uniform. A crucial ingredient is that we can ``approximate'' $f_i$ by nice functions, and that these approximations are essentially undetected by averages of the above form~\cite{MR1844079,MR2415379,GreenTaoLinear}. These nice functions will be averages of $f_1,\ldots,f_m$ on polynomial factors.

A {\em polynomial factor} $\mathcal{B}$ is a partition (sigma-algebra) of $\F_p^n$ defined by a collection of polynomials $P_1,\ldots,P_C$. The atoms of the partitions are
$$
\{x \in \F_p^n: P_1(x)=a(1),\ldots,P_C(x)=a(C)\},
$$
where $a(1),\ldots,a(C) \in \F_p$. The {\em degree} of $\mathcal{B}$ is the maximal degree of $P_1,\ldots,P_C$, and the {\em complexity} of $\mathcal{B}$ is $C$.

In Definition~\ref{def:rank}, to every set of polynomials $\P$ we will assign a natural number $\rank(\P)$. The rank of a single polynomial $P$ is
defined to be $\rank(\{P\})$, and rank of a polynomial factor $\mathcal{B}$ is the rank of the set of polynomials defining it. Roughly speaking, a polynomial of high rank behaves similar to a generic polynomial in the sense that it cannot be simplified to a function of a few polynomials with strictly  lower degrees.

The conditional average of a function $f:\F_p^n \to \D$ over a polynomial factor $\mathcal{B}$, denoted $\E(f|\mathcal{B}):\F_p^n \to \D$, is defined as
$$
\E(f|\mathcal{B})(x) = \Ex_{\{y \in \F_p^n: P_1(y)=P_1(x),\ldots,P_C(y)=P_C(x)\}}[f(y)].
$$
That is, $\E(f|\mathcal{B})(x)$ is the average of $f$ in the atom to which $x$ belongs. The usefulness of these averages is that they may be used to approximate the function $f$. It is known (see Theorem~\ref{thm:decompose_multiple_func})
that given any $d<p$ and $\delta>0$, for every function $f:\F_p^n \to \D$, there exists a polynomial factor $\mathcal{B}$ of degree $d$ and bounded complexity such that
$$
\|f - \E(f|\mathcal{B})\|_{U^{d+1}} \le \delta.
$$
As we shall see, this allows us to replace any arbitrary function $f:\mathbb{F}_p^n \rightarrow \mathbb{D}$ with a structured function $\E(f|\mathcal{B})$ without changing averages $t_{\mathcal{L},\alpha}$ significantly. Note that $\E(f|\mathcal{B})$ is $\mathcal{B}$-measurable, that is, it is constant on the atoms of $\mathcal{B}$. Thus, it will suffice to study averages $t_{\mathcal{L},\alpha}(g)$ where $g$ is a $\mathcal{B}$-measurable function.

As we shall see, we may assume that all the systems of linear forms arising in our proofs are {\em homogeneous}. Informally, a system of linear forms is homogeneous if (maybe after some change of basis) there is a variable that appears with coefficient exactly one in every linear form.
One of the key ingredients in the proof of our main result is the following invariance theorem which is established in~\cite{ComplexityPaper}.

\begin{thm}[Proposition~\ref{prop:invariance}]
Let $\P=\{P_1,\ldots,P_C\},\Q=\{Q_1,\ldots,Q_C\}$ be two collections of polynomials over $\F_p^n$ of degree at most $d<p$ such that $\deg(P_i)=\deg(Q_i)$ for every $1 \le i \le C$. Let $\mathcal{L}=\{L_1,\ldots,L_m\}$ be a homogeneous system of linear forms, and $\Gamma:\F_p^{C} \to \mathbb{D}$ be an arbitrary function. Define $f,g:\F_p^n \to \mathbb{D}$ by
$$
f(x) = \Gamma(P_1(x),\ldots,P_C(x))
$$
and
$$
g(x) = \Gamma(Q_1(x),\ldots,Q_C(x)).
$$
Then for every $\eps>0$, if $\rank(\P),\rank(\Q)$ are large enough, then for all $\alpha \in \{0,1\}^m$,
$$
\left| t_\mathcal{L,\alpha}(f) - t_\mathcal{L,\alpha}(g) \right| \le \eps.
$$
\end{thm}

\subsection{Interior of a set of averages of linear forms}
Another component in the proof of Theorem~\ref{thm:F:main} is an extension of a result of Erd\"os, Lov\'asz, and Spencer~\cite{MR538044} to the setting of additive combinatorics. We show (Theorem~\ref{thm:InteriorLinearForms}) that if $\mathcal{L}_1,\ldots,\mathcal{L}_{\ell}$ are systems of linear forms, then unless there are some trivial obstructions, there exists $N \in \N$ such that the set
$$
\left\{\left(t_{\mathcal{L}_1}(f),\ldots,t_{\mathcal{L}_{\ell}}(f)\right): \ \  f:\F_p^N \to [0,1]\right\} \subset \mathbb{R}^{\ell}
$$
has a nonempty interior. Formally, for the theorem to hold, we require the systems of linear forms to be {\em non-isomorphic} and {\em connected}.

Two systems of linear forms $\mathcal{L'}=\{L'_1,\ldots,L'_m\}$ and $\mathcal{L''}=\{L''_1,\ldots,L''_m\}$ are {\em isomorphic}, if the distributions $(L'_1(\X),\ldots,L'_m(\X))$ and $(L''_1(\X),\ldots,L''_m(\X))$ are identical (after rearranging them if necessary). Note that if $\mathcal{L'},\mathcal{L''}$ are isomorphic then
$$
t_{\mathcal{L'}}(f)=t_{\mathcal{L''}}(f),
$$
for all functions $f:\F_p^n \to \D$.

A system of linear forms $\mathcal{L}=\{L_1,\ldots,L_m\}$ is {\em connected} if it cannot be partitioned as $\mathcal{L}=\mathcal{L}_1 \dotcup \mathcal{L}_2$ where $\span(\mathcal{L}_1) \cap \span(\mathcal{L}_2) = \{\vec{0}\}$. Note that if $\mathcal{L}$ is not connected, then
$$
t_{\mathcal{L}}(f)=t_{\mathcal{L}_1}(f) t_{\mathcal{L}_2}(f),
$$
for all functions $f:\F_p^n \to \D$.

\subsection{Proof for strongly correlation testable families}
Let $\mathcal{D}=\{D_n\}_{n \in \N}$ be a family which is strongly correlation testable. Assume for simplicity of exposition in the proof overview that instead of the original definition of strongly correlation testable families given in Section~\ref{sec:linforms}, the following slightly stronger statement holds\footnote{In the actual proof we use the generalized averages $t_{\mathcal{L}_i,\alpha_i}(f)$.}. There exists a system of linear forms $\mathcal{L}_1,\ldots,\mathcal{L}_{\ell}$ such that for every $\eps>0$, there exists $\delta \in (0,\eps)$ and $n_0 \in \N$, such that the  closures of the following two sets are disjoint:
$$T_\eps := \{(t_{\mathcal{L}_1}(f),\ldots, t_{\mathcal{L}_{\ell}}(f)) | n >n_0,
f:\F_p^n \rightarrow \D, \|f\|_{u(D_n)} \ge \eps\},$$
and
$$S_\eps := \{(t_{\mathcal{L}_1}(f),\ldots, t_{\mathcal{L}_{\ell}}(f)) | n >n_0,
f:\F_p^n \rightarrow \D, \|f\|_{u(D_n)} \le \delta\}.$$
We can furthermore assume that these systems are homogeneous, and there is a number $s<p$ such that  for all $i \in [\ell]$,
$$\left|t_{\mathcal{L}_i}(f) - t_{\mathcal{L}_i}(g)\right| \approx 0,$$
if $f,g:\mathbb{F}_p^n \rightarrow \mathbb{D}$ satisfy $\|f-g\|_{U^{s+1}} \approx 0$. This follows from the assumption that the number of queries is $q \le p$.

Let $t \in \N$ be maximal such the following holds. There exist polynomials $Q_n$ of degree exactly $t$ such that
\begin{itemize}
\item $Q_n$ has noticeable correlation with $D_n$: that is $\liminf_{n \to \infty} \|\exp(Q_n)\|_{u(D_n)}>0$.
\item $Q_n$ has ``large enough rank'' (exact definition is deferred to the actual proof).
\end{itemize}
We can show that $t \le s$: if $Q_n$ are polynomials of degree greater than $s$ and large enough rank, then $\|\exp(Q_n)\|_{U^{s+1}} \approx 0$ and hence $t_{\mathcal{L}_i}(\exp(Q_n)) \approx 0$, for all $i \in [\ell]$. Thus we must have $\|\exp(Q_n)\|_{u(D_n)} \approx 0$.

Let $f_n:\F_p^n \to \D$ be a sequence of functions. We establish the main result (for strongly correlation testable families) by showing that:
$$
\lim_{n \to \infty} \|f_n\|_{u(D_n)}=0 \Longleftrightarrow \lim_{n \to \infty} \|f_n\|_{U^{t+1}}=0.
$$

\paragraph{Sketch of the proof of $\Leftarrow$:}
Assume to the contrary that $\lim_{n \to \infty}\|f_n\|_{U^{t+1}}=0$ but that $\eps:=\liminf_{n \to \infty}\|f_n\|_{u(D_n)}>0$. Let $\eps'>0$ be the $L_{\infty}$ distance between $T_{\eps}$ and $S_{\eps}$ (here we use the fact that their closures are disjoint, hence there is positive distance between the sets). Let $\mathcal{B}_n$ be a polynomial factor of degree $s$ such that $\|f_n - \E(f_n|\mathcal{B}_n)\|_{U^{s+1}} \le \eta$. Choosing $\eta>0$ small enough we can guarantee that $|t_{\mathcal{L}_i}(f_n)-t_{\mathcal{L}_i}(\E(f_n|\mathcal{B}_n))| \le \eps'/2$, hence we must have $\|\E(f_n|\mathcal{B}_n)\|_{u(D_n)} \ge \delta$.

Assume $\mathcal{B}_n$ is defined by polynomials $P_{n,1},\ldots,P_{n,C}$, which we can assume to be of high enough rank. For $\gamma \in \F_p^C$, define $P_{n,\gamma}(x) := \sum \gamma(i) P_{n,i}(x)$. We have $\E(f_n|\mathcal{B}_n)(x)=\sum_{\gamma \in \F_p^C} c_{\gamma} \exp(P_{n,\gamma}(x))$ where $|c_{\gamma}| \le 1$. Using the assumption that $\lim_{n \to \infty}\|f_n\|_{U^{t+1}}=0$, we show that if $\deg(P_{\gamma}) \le t$, then its contribution to the sum is negligible, i.e.  $c_{\gamma}$ can be assumed to be arbitrarily small. Thus, there must exist a polynomial $P_{n,\gamma}$ of degree at least $t+1$ such that $\|\exp(P_{n,\gamma})\|_{u(D_n)} \ge \delta p^{-C}$. As this polynomial has ``large enough rank'', this contradicts the maximality of $t$.

\paragraph{Sketch of the proof of $\Rightarrow$:}
Assume that $\lim_{n \to \infty}\|f_n\|_{U^{t+1}}>0$ (actually we have $\liminf$, which we can replace by an actual limit by Condition $\bf A1$). Thus, there exist polynomials $P_n$ of degree at most $t$ such that
$|\ip{f_n,\exp(P_n)}| \ge \eps>0$. We can assume these polynomials to have ``large enough rank''. Assume first that these
polynomials are of degree exactly $t$. By the definition of $t$, there exist polynomials $Q_n$ of degree $t$ and ``large enough rank'' such that $\|\exp(Q_n)\|_{u(D_n)} \ge \eps'$. We use an  ``invariance'' result (Proposition~\ref{prop:invariance}) to construct a new sequence of functions $f'_n:\F_p^n \to \D$ (essentially replacing $P_n$ with $Q_n$) such that
\begin{itemize}
\item $f'_n$ has correlation with $Q_n$. In fact, this correlation is strong enough to guarantee that $\|f'_n\|_{u(D_n)}\ge \eps''$ for some $\eps''>0$ independent of $n$.
\item Averages $t_{\mathcal{L}_i}$ cannot distinguish $f_n$ from $f'_n$. So, we must also have $\|f_n\|_{u(D_n)} \ge \delta'' >0$.
\end{itemize}
The case where $P_n$ have degrees less than $t$ is reduced to the case of degree $t$. We use the theorem on the interior of a set of averages of linear forms to argue that we can (essentially) tweak $f_n$ slightly, without changing averages $t_{\mathcal{L}_i}$ by much, but such that $f_n$ will have correlation with polynomials of degree exactly $t$ and large enough rank.

\subsection{Proof for correlation testable families}
The proof for the correlation testable families (i.e. with guarantees only for functions $\ff:\F_p^n \to \F_p$) follows  similar steps, albeit slightly more involved. Let $\Dp \subset \mathbb{R}^p$ denote the convex set of probability distributions on $\F_p$. Then correlation testable families can in fact test randomized functions as well, i.e. functions $\ff:\F_p^n \to \Dp$. Once this is established the remainder of the proof follows identical lines to the case of strongly correlation testable families, where we employ the fact that $\Dp$ is a convex set to appropriately define averages such as $\E(\ff|\mathcal{B})$.

\section{Systems of linear forms}
\label{sec:systemLinearForms}
The precise definitions of a linear form, a system of linear forms, and the notations $t_\mathcal{L}(\cdot)$, $t_{\mathcal{L},\alpha}(\cdot)$, and $t^*_{\mathcal{L},\beta}(\cdot)$ are already given in Section~\ref{sec:linforms}.

Let  $\mathcal{L} = \{L_1,\ldots,L_m\}$ be a system of linear forms in $k$ variables, and $\X$ be a random variable taking values uniformly in $(\F_p^n)^k$. 
Note that if $A \subseteq \F_p^n$ and $1_A:\F_p^n \to \{0,1\}$ is the indicator function for $A$, then $t_{\mathcal{L}}(1_A)$ is the probability that $L_1(\X),\ldots,L_m(\X)$ all fall in $A$. Roughly speaking, we say $A \subseteq \F_p^n$ is \emph{pseudorandom} with regards to $\mathcal{L}$, if
$$
t_{\mathcal{L}}(1_A) \approx \left( \frac{|A|}{p^n} \right)^m,
$$
that is if the probability that all $L_1(\X),\ldots,L_m(\X)$ fall in $A$ is close to what we would expect if $A$ was a random subset of $\F_p^n$ of size $|A|$. Let $\alpha = |A|/p^n$ be the relative measure of $A$, and define $f(x):=1_A(x) - \alpha$. We have
$$
t_{\mathcal{L}}(1_A)=t_{\mathcal{L}}(\alpha+f)=\alpha^m + \sum_{S \subseteq [m], S \ne \emptyset} \alpha^{m-|S|} \cdot t_{\{L_i: i \in S\}}(f).
$$
So, a sufficient condition for $A$ to be pseudorandom with regards to $\mathcal{L}$ is that $t_{\{L_i: i \in S\}}(f) \approx 0$ for all nonempty subsets $S \subseteq [m]$. Green and Tao~\cite{GreenTaoLinear} showed that a sufficient condition for this to occur is that $\|f\|_{U^{s+1}}$ is small enough, where $s$ is the {\em Cauchy-Schwarz complexity} of the system of linear forms.

\begin{definition}[Cauchy-Schwarz complexity~\cite{GreenTaoLinear}]
\label{dfn:cauchy-schwarz-complexity}
Let $\mathcal{L}=\{L_1,\ldots,L_m\}$ be a system of linear forms. The {\em Cauchy-Schwarz complexity} of $\mathcal{L}$ is the minimal $s$ such that the following holds. For every $1 \le i \le m$, we can partition $\{L_j\}_{j \in [m] \setminus \{i\}}$ into $s+1$ subsets, such that $L_i$ does not belong to  the linear span of any of the subsets.
\end{definition}
The reason for the term {\em Cauchy-Schwarz complexity} is the following lemma due to Green and Tao~\cite{GreenTaoLinear} whose proof is based on a clever iterative application of the Cauchy-Schwarz inequality.

\begin{lemma}[{\cite{GreenTaoLinear}, See also~\cite[Theorem 2.3]{MR2578471}}]
\label{lemma:cauchy-schwarz-lemma}
Let $f_1,\ldots,f_m:\F_p \to \D$. Let $\mathcal{L}=\{L_1,\ldots,L_m\}$ be a system of $m$ linear forms in $k$ variables of Cauchy-Schwarz complexity $s$. Then
$$
\left| \Ex_{\X \in (\F_p^n)^k} \left[ \prod_{i=1}^m f_i (L_i(\X)) \right]\right| \le \min_{1 \le i \le m} \|f_i\|_{U^{s+1}}.
$$
\end{lemma}
Note that the Cauchy-Schwarz complexity of any system of $m$ linear forms in which any two linear forms are linearly independent (i.e. one is not a multiple of the other) is at most $m-2$, since we can always partition $\{L_j\}_{j \in [m] \setminus \{i\}}$ into the $m-1$ singleton subsets.

%
%

\subsection{The true complexity of linear forms}
The Cauchy-Schwarz complexity of $\mathcal{L}$ gives an upper bound on $s$, such that if $\|f\|_{U^{s+1}}$ is small enough for some function $f:\mathbb{F}_p^n \rightarrow \mathbb{D}$, then $f$ is pseudorandom with regards to $\mathcal{L}$. Gowers and Wolf~\cite{MR2578471} defined the {\em true complexity} of a system of linear forms as the minimal $s$ such that the above condition holds for all $f:\mathbb{F}_p^n \rightarrow \mathbb{D}$.
\begin{definition}[True complexity~\cite{MR2578471}]
\label{def:trueComplexity}
Let $\mathcal{L}=\{L_1,\ldots,L_m\}$ be a system of linear forms over $\F_p$.
The true complexity of $\mathcal{L}$ is the smallest $d \in \N$ with the following property. For every $\eps>0$,
there exists  $\delta>0$ such that if $f : \F_p^n \rightarrow \D$ is any function
with $\|f\|_{U^{d+1}} \le \delta$, then
$$\left| t_{\mathcal{L}}(f) \right|\le  \eps.$$
\end{definition}
An obvious bound on the true complexity is the Cauchy-Schwarz complexity of the system. However, there are cases where this is not tight. Gowers and Wolf~\cite{gowers-wolf-2010} characterized the true complexity of systems of linear forms, assuming the field is not too small. For a linear form $L \in \F_p^m$, let $L^{k} \in \F_p^{m^k}$ be the $k$-tensor power of $L$. That is, if $L=(\lambda_1,\ldots,\lambda_m)$, then
$$
L^{k} = \left(\prod_{j=1}^k \lambda_{i_j}: i_1,\ldots,i_k \in [m]\right) \in \mathbb{F}_p^{m^k}.
$$

\begin{theorem}[Characterization of the true complexity of linear systems, Theorem 6.1 in~\cite{gowers-wolf-2010}]
\label{thm:true_compltexity_characterization}
Let $\mathcal{L}=\{L_1,\ldots,L_m\}$ be a system of linear forms over $\F_p^n$ of Cauchy-Schwarz complexity $s$,
and assume that $s \le p$. The true complexity of $\mathcal{L}$ is the minimal $d$ such that $L_1^{d+1},\ldots,L_m^{d+1}$ are linearly independent over $\F_p$.
\end{theorem}

In~\cite{ComplexityPaper} the authors proved the following strengthening of Theorem~\ref{thm:true_compltexity_characterization}. 
\begin{theorem}[\cite{ComplexityPaper}]
\label{thm:strong_independence}
Let $\mathcal{L}=\{L_1,\ldots,L_m\}$ be a system of linear forms  of Cauchy-Schwarz complexity at most $p$. Let $d \ge 0$, and assume that $L_1^{d+1}$ is not in the linear span of $L_2^{d+1},\ldots,L_m^{d+1}$. Then for every $\eps>0$, there exists $\delta>0$ such that for any functions $f_1,\ldots,f_m:\F_p \to \D$ with $\|f_1\|_{U^{d+1}} \le \delta$, we have
$$
\left| \Ex_{\X \in (\F_p^n)^k} \left[\prod_{i=1}^m f_i(L_i(\X)) \right] \right| \le \eps,
$$
where $\X \in (\F_p^n)^k$ is uniform.
\end{theorem}
The following theorem immediately follows from Theorem~\ref{thm:strong_independence}. 
\begin{theorem}[\cite{ComplexityPaper}]
\label{thm:avg-approx-funcs}
Let $\mathcal{L}=\{L_1,\ldots,L_m\}$ be a system of linear forms of true complexity $d$ and Cauchy-Schwarz complexity at most $p$. Then for every $\eps>0$, there exists $\delta>0$ such that the following holds. Let $f_i,g_i:\F_p^n \to \D$
for $1 \le i \le m$ be functions such that $\|f_i - g_i\|_{U^{d+1}} \le \delta$. Then
$$
\left|\Ex_{\X} \left[\prod_{i=1}^m f_i(L_i(\X))\right] - \Ex_{\X} \left[\prod_{i=1}^m g_i(L_i(\X))\right]\right| \le \eps,
$$
where $\X \in (\F_p^n)^k$ is uniform.
\end{theorem}

\subsection{Equivalence of systems of linear forms}

Let $S$ be a subset of  $\F_p^k$. Here we think of $S$ as a linear structure, and we are interested in the number of affine copies of $S$ in a set $A \subseteq \F_p^n$. Pick a linear transformation $T \in \Lin(\F_p^k, \F_p^n)$ uniformly at random, and also independently and uniformly a random element $X \in \F_p^n$. Then the density of the structure $S$ in $A$ is the probability that $X + T(x)  \in A$ for every $x \in S$.  Note that a uniform random $T \in \Lin(\F_p^k, \F_p^n)$ can be defined by mapping each standard vector $\e_i \in \F_p^k$ uniformly and independently to a point in $\F_p^n$. Hence if the elements of a linear structure $S$ are $(\lambda_{i,1},\ldots,\lambda_{i,k}) \in \F_p^k$ where $1 \le i \le m$, then the density of the structure $S$ in $A$ is the probability that for every $1 \le i \le m$, $X + \sum_{j=1}^k \lambda_{i,j} Y_j \in A$, where $X,Y_1,\ldots,Y_m$ are i.i.d. random variables taking values uniformly in $\F_p^n$. So if we define the system of linear forms $\mathcal{L}=\{L_1,\ldots,L_m\}$ by letting
\begin{equation}
L_i = (1,\lambda_{i,1},\ldots,\lambda_{i,k}),
\end{equation}
for $1 \le i \le m$, then $t_\mathcal{L}(1_A)$ gives the density of the structure $S$ in a set $A \subseteq \F_p^n$.
Note that since for a fixed $c \in \F_p^n$, a uniform random variable $X$ has the same distribution as $X+c$, for this system of linear forms $\mathcal{L}$, the distribution of $(L_1(\X),\ldots,L_m(\X))$ is the same as the distribution of  $(L_1(\X)+c,\ldots,L_m(\X)+c)$.
\begin{definition}[Homogeneous linear forms]
\label{def:homogen}
A system of linear forms $\mathcal{L} = \{L_1,\ldots,L_m\}$  in $k$ variables is called \emph{homogeneous}, if for a uniform random variable $\X \in (\F_p^n)^k$, and every fixed $c \in \F_p^n$,  $(L_1(\X),\ldots,L_m(\X))$ has the same distribution as $(L_1(\X)+c,\ldots,L_m(\X)+c)$.
\end{definition}
For example, the systems of linear forms used in the definition of Gowers uniformity norms are homogeneous.

We wish to identify two systems of linear forms $\mathcal{L}_0 = \{L_1,\ldots,L_m\}$ in $k_0$ variables, and $\mathcal{L}_1=\{L_1',\ldots,L_m'\}$ in $k_1$ variables,  if $(L_1(\X),\ldots,L_m(\X))$ has the same distribution as $(L'_1(\Y),\ldots,L'_m(\Y))$ where $\X$ and $\Y$ are uniform random variables taking values in $(\F_p^n)^{k_0}$ and $(\F_p^n)^{k_1}$, respectively.  The distribution of $(L_1(\X),\ldots,L_m(\X))$ depends exactly on the linear dependencies between $L_1,\ldots,L_m$, and two systems of linear forms lead to the same distributions  if and only if they have the same linear dependencies.

\begin{definition}[Isomorphic linear forms]
\label{def:linearIsomorphic}
Two systems of linear forms $\mathcal{L}_0 $ and $\mathcal{L}_1$ are \emph{isomorphic} if and only if there exists
a \emph{bijection} from $\mathcal{L}_0$ to $\mathcal{L}_1$  that can be extended to an \emph{invertible} linear transformation $T: {\rm span}(\mathcal{L}_0) \rightarrow {\rm span}(\mathcal{L}_1)$.
\end{definition}

Note that if $\mathcal{L} = \{L_1,\ldots,L_m\}$ is a homogeneous system of linear forms, then $(L_1(\X),\ldots,L_m(\X))$ has the same distribution as $(L_1(\X)+Y,\ldots,L_m(\X)+Y)$, where $Y$ is a uniform random variable taking values in $\F_p^n$ and is independent of $\X$. We conclude with the following trivial observation.

\begin{observation}
\label{obs:canonical}
Every homogeneous system of linear forms is isomorphic to a system of linear forms in which there is a variable that appears with coefficient exactly one in every linear form.
\end{observation}

Consider a system of linear forms $\mathcal{L}$ in $\F_p^k$, and note that for every $f:\F_p^n \to \C$,
\begin{equation}
t_\mathcal{L}(f) = \Ex \left[\prod_{L \in \mathcal{L}} f(T(L))\right],
\end{equation}
where $T$ is a random variable taking values uniformly in  $\Lin(\span(\mathcal{L}), \F_p^n)$. Suppose that there exists a non-trivial subset $S \subseteq \mathcal{L}$ such that
$${\rm span}(S) \cap  {\rm span}(\mathcal{L} \setminus S) = \{\vec{0}\}.$$
Then for every $f:\F_p^n \to \C$ we have
$$t_\mathcal{L}(f) = t_S (f) t_{\mathcal{L}\setminus S}(f).$$
This leads to the following definition.
\begin{definition}[Connected linear forms]
\label{def:Connected}
A system of linear forms $\mathcal{L}$ is called \emph{connected}, if for every non-trivial subset $S \subsetneq \mathcal{L}$, we have
$${\rm span}(S) \cap  {\rm span}(\mathcal{L}\setminus S) \neq \{\vec{0}\}.$$
\end{definition}

\section{Higher-order Fourier analysis}
\label{sec:highFourier}
The characters of $\F_p^n$ are exponentials of linear polynomials; that is for $\alpha \in \F_p^n$, the corresponding character is defined
as $\chi_\alpha(x) = \exp(\sum_{i=1}^n \alpha_i x_i)$. In higher-order Fourier analysis, the linear polynomials $\sum \alpha_i x_i$ are replaced by higher degree polynomials, and one would like to express a function $f:\F_p^n \to \C$ as a linear combination of the functions $\exp(P)$, where $P$ is a polynomial of a certain degree.

\emph{Polynomials:} Consider a function $\ff: \F_p^n \rightarrow \F_p$. For an element $y \in \F_p^n$, define the \emph{derivative} of $\ff$ in the direction $y$ as $\Delta_y \ff(x)= \ff(x+y)-\ff(x)$. Inductively we define $\Delta_{y_1,\ldots,y_k}\ff=\Delta_{y_k}(\Delta_{y_1,\ldots,y_{k-1}} \ff)$, for directions $y_1,\ldots,y_k \in \F_p^n$. We say that $\ff$ is a \emph{polynomial of degree at most $d$}, if for every $y_1,\ldots,y_{d+1} \in \F_p$, we have $\Delta_{y_1,\ldots,y_{d+1}} \ff \equiv 0$. The set of polynomials of degree at most $d$ is a vector space over $\F_p$ which we denote by ${\rm Poly}_d(\F_p^n)$. It is easy to see that the set of \emph{monomials} $x_1^{i_1} \ldots x_n^{i_n}$ where $0 \le i_1,\ldots,i_n < p$ and $\sum_{j=1}^n i_j \le d$ form a basis for ${\rm Poly}_d(\F_p^n)$. So every polynomial $P \in {\rm Poly}_d(\F_p^n)$ is of the from $P(x):=\sum c_{i_1,\ldots,i_n} x_1^{i_1} \ldots x_n^{i_n}$, where the sum is over all $1 \le i_1,\ldots,i_n < p$ with $\sum_{j=1}^n i_j \le d$, and $c_{i_1,\ldots,i_n}$ are elements of $\F_p$. The \emph{degree} of a polynomial $P:\F_p^n \rightarrow \F_p$, denoted by $\deg(P)$, is the smallest $d$ such that $P \in {\rm Poly}_d(\F_p^n)$. A polynomial  $P$ is called \emph{homogeneous}, if all monomials with non-zero coefficients in the expansion of $P$ are of degree exactly $\deg(P)$.

For a function $f:\F_p^n \rightarrow \C$, and a direction $y \in \F_p^n$ define the \emph{multiplicative derivative} of $f$ in the direction of $y$ as $\tilde\Delta_{y}f(x)=f(x+y)\overline{f(x)}$. Inductively we define $\tilde{\Delta}_{y_1,\ldots,y_k}f=\tilde{\Delta}_{y_k}(\tilde{\Delta}_{y_1,\ldots,y_{k-1}} f)$, for directions $y_1,\ldots,y_k \in \F_p^n$. Note that
for every $\ff: \mathbb{F}_p^n \rightarrow \mathbb{F}_p$, we have
$$\tilde{\Delta}_y \exp(\ff) = \exp(\Delta_y \ff).$$
This shows that if $\ff \in {\rm Poly}_d(\F_p^n)$, then
\begin{equation}
\label{eq:polyphase}
\tilde{\Delta}_y \exp(\ff)=1.
\end{equation}
Note that one can express Gowers uniformity norms using multiplicative derivatives:
$$\|f\|_{U^k}^{2^k} = \Ex\left[ \tilde{\Delta}_{Y_1,\ldots,Y_k}f(X)\right],$$
where $X,Y_1,\ldots,Y_k$ are independent random variables taking values in $\F_p^n$ uniformly. This together with (\ref{eq:polyphase}) show that every polynomial $\gg$ of degree at most $d$ satisfies $\|\exp(g)\|_{U^{d+1}}=1$.

Many basic properties of  Gowers uniformity norms are implied by the Gowers-Cauchy-Schwarz inequality which is first proved in~\cite{MR1844079} by iterated applications of the classical Cauchy-Schwarz inequality.
\begin{lemma}[Gowers-Cauchy-Schwarz]
\label{lem:GowersCauchyShwarz}
Let $G$ be a finite Abelian group, and consider a family of functions $f_S:G \rightarrow \C$, where $S \subseteq [k]$. Then
\begin{equation}
\left| \Ex \left[\prod_{S \subseteq [k]} \mathcal{C}^{k-|S|} f_S(X + \sum_{i \in S} Y_i) \right]\right|\le \prod_{S \subseteq [k]} \|f_S\|_{U^k},
\end{equation}
where $X,Y_1,\ldots,Y_k$ are independent random variables taking values in $G$ uniformly at random.
\end{lemma}
A simple application of Lemma~\ref{lem:GowersCauchyShwarz} is the following. Consider an arbitrary function $f:G \rightarrow \C$.
Setting $f_\emptyset :=f$ and $f_S:=1$ for every $S \neq \emptyset$ in Lemma~\ref{lem:GowersCauchyShwarz}, we obtain
\begin{equation}
\label{eq:avgVsGowers}
|\Ex[f(X)] | \le \|f\|_{U^k}.
\end{equation}
Equation~(\ref{eq:avgVsGowers}) in particular shows that if $f,g:\F_p^n \to \C$,
then one can bound their inner product  with Gowers uniformity norms of $f\overline{g}$:
\begin{equation}
\label{eq:GowersCorrelation}
|\ip{f,g}| \le \|f\overline{g}\|_{U^k}.
\end{equation}
Consider an arbitrary $f:\F_p^n \to \C$ and a polynomial $\gg$ of degree at most $d$. Let $g=\exp(\gg)$. Then for every $y_1,\ldots,y_{d+1} \in \F_p^n$, we have
$$\tilde{\Delta}_{y_1,\ldots,y_{d+1}} (f g) = (\tilde{\Delta}_{y_1,\ldots,y_{d+1}} f) (\tilde{\Delta}_{y_1,\ldots,y_{d+1}}g) =\tilde{\Delta}_{y_1,\ldots,y_{d+1}} f,$$
which in turn implies that $\|fg\|_{U^{d+1}}=\|f\|_{U^{d+1}}$. Combining this with (\ref{eq:GowersCorrelation}), we conclude that the correlation of $f$ with any polynomial of degree at most $d$ is a lower-bound for $\|f\|_{U^{d+1}}$:
\begin{equation}
\label{eq:GowersCorrelationII}
\|f\|_{u(\mathrm{Poly}_d)} \le \|f\|_{U^{d+1}}.
\end{equation}
This provides us with a ``direct theorem'' for the $U^{d+1}$ norm: If $\|f\|_{u(\mathrm{Poly}_d)} \ge \delta$, then $\|f\|_{U^{d+1}} \ge \delta$. Recently Bergelson, Tao, and Ziegler~\cite{bergelson-2009,tao-2008} established the corresponding inverse theorem in the high characteristic case.
\begin{theorem}[\cite{bergelson-2009,tao-2008}]
\label{thm:inverse}
If $1\le d<p$, then there exists a function $\delta:(0,1] \rightarrow (0,1]$ such that for every $f:\F_p^n \rightarrow \mathbb{D}$, and $\eps>0$,
\begin{itemize}
\item \emph{Direct theorem:} If $\|f\|_{u(\mathrm{Poly}_d)} \ge \eps$, then $\|f\|_{U^{d+1}} \ge \eps$.
\item \emph{Inverse theorem:} If $\|f\|_{U^{d+1}} \ge \eps$, then $\|f\|_{u(\mathrm{Poly}_d)} \ge \delta(\eps).$
\end{itemize}
\end{theorem}

An important application of the inverse theorems is that they imply  ``decomposition theorems''. Roughly speaking these results say that under appropriate conditions, a function $f$ can be decomposed as $f_1+f_2$, where $f_1$ is ``structured'' in some sense that enables one to handle it easily, while $f_2$ is ``quasi-random'' meaning that it shares certain properties with a random function, and can be discarded as random noise (see \cite{GowersSurvey}). In the following we will discuss decomposition theorems that follow from Theorem~\ref{thm:inverse}, but first we need to define the polynomial factors on $\F_p^n$.
\begin{definition}[Polynomial factors~\cite{GreenTaoFiniteFields}]
\label{def:polyFactor}
Let $p$ be a fixed prime. Let $P_1,\ldots,P_C \in {\rm Poly}_d(\F_p^n)$. The sigma-algebra on $\F_p^n$ whose atoms are $\{x \in \F_p^n : P_1(x)=a(1),\ldots,P_C(x)=a(C)\}$ for all $a \in \F_p^C$ is called a \emph{polynomial factor of degree at most $d$ and complexity at most $C$}.
\end{definition}
Let $\mathcal{B}$ be a polynomial factor defined by $P_1,\ldots,P_C$. For $f:\F_p^n \to \C$, the conditional expectation of $f$ with respect to $\mathcal{B}$, denoted $\E(f|\mathcal{B}):\F_p^n \to \C$, is
$$
\E(f|\mathcal{B})(x) = \Ex_{\{y \in \F_p^n: P_1(y)=P_1(x),\ldots,P_C(y)=P_C(x)\}}[f(y)].
$$
That is, $\E(f|\mathcal{B})$ is constant on every atom of $\mathcal{B}$, and this constant is the average value that $f$ attains on this atom. A function $g:\F_p^n \to \C$ is $\mathcal{B}$-measurable, if it is constant on every atom of $\mathcal{B}$. Equivalently, we can write $g$ as $g(x)=\Gamma(P_1(x),\ldots,P_C(x))$ for some function $\Gamma:\F_p^C \to \C$. The following claim is quite useful, although its proof is immediate and holds for every sigma-algebra.
\begin{observation}
\label{obs:project_factor}
Let $f:\F_p^n \to \C$. Let $\mathcal{B}$ be a polynomial factor defined by polynomials $P_1,\ldots,P_C$.
Let $g:\F_p^n \to \C$ be any $\mathcal{B}$-measurable function. Then
$$
\ip{f,g} = \ip{\E(f|\mathcal{B}),g}.
$$
\end{observation}

\begin{definition}[Bias]
The bias of a polynomial $P \in {\rm Poly}_d(\F_p^n)$ is defined as
$$
\bias(P) := \bias(\exp(P)) = |\Ex_{X \in \F_p^n}[\exp(P(X))]|.
$$
\end{definition}
We shall refine the set of polynomials $\{P_1,\ldots,P_t\}$ to obtain a new set of polynomials with the desired ``approximate orthogonality'' properties. This will be achieved through the notion of the \emph{rank} of a set of polynomials.
\begin{definition}[Rank~\cite{GreenTaoFiniteFields}]
\label{def:rank}
We say a set of polynomials $\P=\{P_1,\ldots,P_t\}$ is of rank greater than $r$, and denote this by $\rank(\P) > r$, if the following holds. For any non-zero $\alpha=(\alpha_1,\ldots,\alpha_t) \in \F_p^t$, define $P_{\alpha}(x) := \sum_{j=1}^t \alpha_j P_j(x)$. For $d := \max \{\deg(P_j):\alpha_j \ne 0\}$, the polynomial $P_{\alpha}$ cannot be expressed as a function of $r$ polynomials of degree at most $d-1$. More precisely, it is not possible to find $r$ polynomials $Q_1,\ldots,Q_r$ of degree at most $d-1$, and a function $\Gamma:\F_p^r \to \F_p$ such that
$$
P(x)=\Gamma(Q_1(x),\ldots,Q_r(x)).
$$
The rank of a single polynomial $P$ is defined to be $\rank(\{P\})$.
\end{definition}
The \emph{rank} of a polynomial factor is the rank of the set of polynomials defining it. 
{
\begin{remark}\label{rankrandom}
Consider integers $d,r \ge 1$, and let $P$ be a randomly and uniformly chosen homogeneous polynomial of degree $d$ over $\F_p^n$. There are at least $p^{{n \choose d}}$ such polynomials while the number of degree $d$ polynomials of rank at most $r$ is bounded from above by $p^{p^r} p^{r n^{d-1}}$. It follows that for sufficiently large $n$, with high probability $\rank(P)>r$. 
\end{remark}}

%
%
The following theorem due to Kaufman and Lovett~\cite{kaufman-lovett}  connects the notion of the rank to the bias of a polynomial. It was proved first by Green and Tao~\cite{GreenTaoFiniteFields} for the case $d<p$, and then extended in~\cite{kaufman-lovett} for the general case.
\begin{theorem}[Regularity~\cite{kaufman-lovett}]
\label{thm:regularity}
Fix $p$ prime and $d \ge 1$. There exists a function $r_{p,d}:(0,1] \rightarrow \mathbb{N}$ such that the following holds. If $P:\F_p^n \rightarrow \F_p$ is a polynomial of degree at most $d$ with $\bias(P) \ge \eps$, then $\rank(P) \le r_{p,d}(\eps)$.
\end{theorem}

\begin{remark}\label{rem:HighRankSmallNorm}
It follows from Theorems~\ref{thm:inverse} and~\ref{thm:regularity} that if a polynomial $P$ of degree $k$ is of sufficiently large rank, then $\|\exp(P)\|_{U^{k}} \le \epsilon$. Indeed, otherwise by Theorem~\ref{thm:inverse} there exists a polynomial $Q$ of degree at most $k-1$ such that $\bias(P-Q) \ge \delta(\epsilon) >0$. Theorem~\ref{thm:regularity} then provides a bound on the rank of $P-Q$. Since $\deg(Q)<\deg(P)$ this implies a bound on the rank of $P$. Also in the opposite direction, if the rank of $P$ is $r$, then there is a polynomial $Q$ of degree at most $k-1$ such that $\bias(P-Q) \ge \epsilon(r)>0$. Now the direct part of Theorem~\ref{thm:inverse} implies that $\|\exp(P)\|_{U^{k}} \ge \epsilon(r)$.
\end{remark}

The following theorem goes back to the work of Green and Tao~\cite{MR2415379} (See also~\cite{ComplexityPaper}).
\begin{theorem}[Strong Decomposition Theorem - multiple functions]
\label{thm:decompose_multiple_func}
Let $p$ be a fixed prime, $0\le d<p$ and $m$ be integers, and $\delta>0$. Let $\mathcal{B}_0$ be a polynomial factor of degree at most $d$ and complexity $C_0$, and let $r:\N \rightarrow \N$ be an arbitrary growth function, and suppose that $n>n_0(p,d,\delta,m,r(\cdot),C_0)$ is sufficiently large. Given every set of functions $f_1,\ldots,f_m:\F_p^n \rightarrow \mathbb{D}$, there exists a decomposition of each $f_i$ as
$$
f_i=h_{i}+h'_{i},
$$
such that
$$\mbox{$h_{i}:=\E(f_i|\mathcal{B})$, \qquad $\|h'_{i}\|_{U^{d+1}} \le \delta$,}$$
where  $\mathcal{B}$ is a polynomial factor that refines $\mathcal{B}_0$ and it is of degree at most $d$, complexity  $C \le C_{\max}$ (where $C_{\max}$ depends on $p,d,\delta,m,r(\cdot),C_0$) and rank at least $r(C)$. Furthermore, if $\mathcal{B}_0$ is of rank at least $r_0(p,d,\delta,m,r(\cdot),C_0)$ then we can assume that $\mathcal{B}$ contains all the polynomials in $\mathcal{B}_0$.
\end{theorem}

\begin{remark}
In Theorem~\ref{thm:decompose_multiple_func} we can assume that  $\mathcal{B}$ is defined by homogeneous polynomials. Of course if we also want $\mathcal{B}$ to include to all the polynomials of $\mathcal{B}_0$, then we need to require that the polynomials in $\mathcal{B}_0$ are homogeneous as well. 
\end{remark}

The following proposition, proved in~\cite{ComplexityPaper}, is one of the key ingredients in the proof of our main result, Theorem~\ref{thm:F:main}.

\begin{proposition}[An invariance result~\cite{ComplexityPaper}]
\label{prop:invariance}
Let $p$ be a fixed prime.
There exists a function $r:\N \times (0,1] \to \N$ such that the following holds. Let $\P=\{P_1,\ldots,P_k\},\Q=\{Q_1,\ldots,Q_k\}$ be two collections of polynomials over $\F_p^n$ of degree at most $d<p$ such that $\deg(P_i)=\deg(Q_i)=d_i$ for every $1\le i \le k$. Let $\mathcal{L}=\{L_1,\ldots,L_m\}$ be a system of linear forms, and $\Gamma:\F_p^{k} \to \mathbb{D}$ be an arbitrary function. Define $f,g:\F_p^n \to \mathbb{D}$ by
$$
f(x) = \Gamma(P_1(x),\ldots,P_k(x))
$$
and
$$
g(x) = \Gamma(Q_1(x),\ldots,Q_k(x)).
$$
Then for every $\eps>0$, if $\rank(\P),\rank(\Q) > r(p,d,\eps)$,  we have
$$
\left| t_\mathcal{L}(f) - t_\mathcal{L}(g) \right| \le 2 p^{mk} \cdot \eps,
$$
provided that at least one of the following two conditions hold:
\begin{enumerate}
\item[(i)] The polynomials $P_1,\ldots,P_k$ and $Q_1,\ldots,Q_k$ are homogeneous.
\item[(ii)] The system of linear forms $\mathcal{L}$ is homogeneous.
\end{enumerate}
\end{proposition}

\section{Characterization of strongly correlation testable properties}
\label{sec:D:testing_correlation}

Consider a family ${\mathcal D}:=\{D_n\}_{n \in \mathbb{N}}$ where $D_n$ is a set of functions from $\F_p^n$ to $\D$.
We recall some basic definitions from the introduction. The correlation of a function $f:\F_p^n \to \D$ with $D_n$ is
$$
\|f\|_{u(D_n)} = \sup_{g \in D_n} |\ip{f,g}|.
$$
Given a function $f:\F_p^n \rightarrow \D$ and a system of linear forms $\mathcal{L}=\{L_1,\ldots,L_m\}$ in $k$ variables, recall that the average of $f$ over $\mathcal{L}$, with conjugations $\alpha \in \{0,1\}^m$, is
$$
t_{\mathcal{L},\alpha}(f) = \Ex_{\X \in (\F_p^n)^k}\left[ \prod_{i=1}^m \conj^{\alpha_i} f(L_i(\X))\right]
$$
where $\conj$ is the conjugation operator. A family $\mathcal{D}$ is said to be correlation testable with linear forms if there exists a set of linear forms $\mathcal{L}_1,\ldots,\mathcal{L}_{\ell}$ along with conjugations $\alpha_1,\ldots,\alpha_{\ell}$, such that the collection of averages $(t_{\mathcal{L}_1,\alpha_1}(f),\ldots,t_{\mathcal{L}_{\ell},\alpha_{\ell}}(f))$ allows to distinguish whether $f$ has noticeable correlation with $D_n$ or negligible correlation with $D_n$. The {\em} true complexity (Cauchy-Schwarz complexity) of $\mathcal{D}$ is the maximal true complexity (Cauchy-Schwarz complexity) of $\{\mathcal{L}_i\}_{i=1,\ldots,\ell}$.

\begin{definition}[Strongly correlation testable properties]
\label{def:D:GeneralcorrTestable}
A family $\mathcal{D}=\{D_n\}_{n \in \N}$ is {\em strongly correlation testable} by linear systems with true complexity $d$ and Cauchy-Schwarz complexity $s$, if the following holds. For every $\eps>0$, there exist $\delta \in (0,\eps)$, $n_0 \in \N$, and systems of homogeneous linear forms
$\mathcal{L}_1,\ldots,\mathcal{L}_{\ell}$ in $m_{1},\ldots,m_{\ell}$ variables, respectively, where each system has
true complexity at most $d$ and Cauchy-Schwarz complexity at most $s$, along with conjugations $\alpha_1 \in \{0,1\}^{m_1},\ldots,\alpha_{\ell} \in \{0,1\}^{m_{\ell}}$ such that the closures of the following two sets are disjoint:
$$
T_{\epsilon}=\left\{\left(t_{\mathcal{L}_1,\alpha_1}(f),\ldots,t_{\mathcal{L}_k,\alpha_k}(f)\right): f:\F_p^n \to \D,
n \ge n_0, \|f\|_{u(D_n)} \ge \eps \right\}
$$
and
$$
S_{\epsilon}=\left\{\left(t_{\mathcal{L}_1,\alpha_1}(f),\ldots,t_{\mathcal{L}_k,\alpha_k}(f)\right): f:\F_p^n \to \D,
n \ge n_0, \|f\|_{u(D_n)} \le \delta \right\}.
$$
\end{definition}
As we mentioned earlier we use the notation $S_\eps$ instead of $S_\delta$ to emphasis the fact that both sets $T_\eps$ and $S_\eps$ are defined according to the parameter $\eps$.

A system $\mathcal{D}=\{D_n\}_{n \in \N}$ where $D_n$ is a set of functions from $\F_p^n$ to $\D$ is {\em consistent} if $D_n \subseteq D_{n+1}$, where we identify $\F_p^n$ with the subspace $\F_p^n \times \{0\}$ of $\F_p^{n+1}$. In this section, we prove the following theorem.

\begin{theorem}[Main theorem: strongly correlation testable functions]
\label{thm:D:main}
Consider a consistent family $\mathcal{D}=\{D_n\}_{n \in \N}$ .
If $\mathcal{D}$ is strongly correlation testable with true complexity $d$ and Cauchy-Schwarz complexity $s<p$, then there exists $0 \le t \le d$ such that the following holds. Let $(f_n:\F_p^n \rightarrow \D)_{n \in \N}$ be a sequence of functions. Then
$$
\lim_{n \to \infty} \|f_n - \Ex[f_n]\|_{u(D_n)}=0 \Longleftrightarrow\lim_{n \to \infty} \|f_n - \Ex[f_n]\|_{U^{t+1}}=0.
$$
\end{theorem}

Define a set $\mathcal{S} \subseteq \N$ to be the set of all degrees $k \ge 1$ for which the following holds. For every growth function $r:\N \rightarrow \N$, there exists $n_0,a \in \N$, such that for every $n \ge n_0$
there exist a polynomial $P_n$ over $\F_p^n$ of degree exactly $k$ such that
\begin{enumerate}
\item[(i)] $\|\exp(P_n)\|_{u(D_n)} \ge 1/a$;
\item[(ii)] $\rank(P_n)>r(a)$.
\end{enumerate}

\begin{claim}
\label{claim:D:one_in_S}
Unless all functions in $\mathcal{D}$ are constant functions, we have $1 \in \mathcal{S}$.
\end{claim}
\begin{proof}
Let $g_{n_0} \in D_{n_0}$ be a nonconstant function. There must exist a nonzero Fourier coefficient $\alpha \in \F_p^{n_0}$ such that
$$
\widehat{g_{n_0}}(\alpha) = \eta \ne 0.
$$
Since we assume that the family $\mathcal{D}$ is consistent, for every $n \ge n_0$, the function
$$
g_n(x(1),\ldots,x(n)) = g_{n_0}(x(1),\ldots,x(n_0))
$$
belongs to $D_n$. Let $P_n(x)=\sum_{i=1}^{n_0} \alpha(i) x(i)$ be a linear function. For every nonzero linear function we have $\rank(P_n)=\infty$. By construction, each linear function $P_n$ has correlation with $D_n$,
$$
\|\exp(P_n)\|_{u(D_n)} \ge |\ip{\exp(P_n),g_n}|=|\eta|>0.
$$
\end{proof}

The case where all functions in $\mathcal{D}$ are constants is easy to analyze, as in this case we have that
$$
\|f\|_{u(D_n)} = \|f\|_{U^1},
$$
for all function $f:\F_p^n \to \D$. Thus, in the sequel we assume that $\mathcal{D}$ contains at least one non-constant function, and hence $1 \in \mathcal{S}$.

\begin{claim}
\label{claim:D:S_bound}
$\mathcal{S} \subseteq \{1,\ldots,s\}$.
\end{claim}

\begin{proof}
Let $k>s$, and assume to the contrary that for every growth function $r:\N \to \N$, there exists $a,n_0 \in \N$,
such that for every $n \ge n_0$, there exists a polynomial $P_n$ of degree $k$ with $\|\exp(P_n)\|_{u(D_n)} \ge 1/a$ and $\rank(P_n) \ge r(a)$. We will show that in this case, the closures of $T_{1/a}$ and $S_{1/a}$ are not disjoint for all $a \in \N$, which will yield a contradiction. Assume to the contrary that they are disjoint for every $a \in \N$. Then for every $a \in \N$, there exists a minimal distance $\mu(a)>0$, such that for every $z' \in T_{1/a}$ and $z'' \in S_{1/a}$ we have
\begin{equation}
\label{eq:D:S_bound:min_dist}
\|z' - z''\|_{\infty} \ge \mu(a).
\end{equation}

Note that even though $a$ is a function of $r$, by choosing $r$ properly, we can guarantee that $r(a)$ is larger than any given constant. It follows form Remark~\ref{rem:HighRankSmallNorm} that if we choose the rank bound $r(a)$ to be large enough, we can guarantee that
\begin{equation}
\label{eq:D:S_bound:U_small}
\|\exp(P_n)\|_{U^{s+1}} \le \mu(a)/2.
\end{equation}

We first note that $0^{\ell}=(0,\ldots,0)$ is in $S_{1/a}$ for every $a \in \N$, since for $f \equiv 0$ we have $t_{\mathcal{L},\alpha}(f)=0$. Combining this with~\eqref{eq:D:S_bound:min_dist} we get that for every sequence of functions $f_n:\F_p^n \to \D$ with $\liminf_{n \to \infty} \|f_n\|_{u(D_n)} \ge 1/a$, we have
\begin{equation}
\label{eq:D:S_bound:t_lower_bound}
\liminf_{n \to \infty} \|(t_{\mathcal{L}_1,\alpha_1}(f_n),\ldots,t_{\mathcal{L}_{\ell},\alpha_{\ell}}(f_n))\|_{\infty} \ge \mu(a).
\end{equation}
Consider now the polynomials $P_n$. By Lemma~\ref{lemma:cauchy-schwarz-lemma}, since each system $\mathcal{L}_i$ has Cauchy-Schwarz complexity at most $s<p$, we have
$$
|t_{\mathcal{L}_i,\alpha_i}(\exp(P_n))| \le \|\exp(P_n)\|_{U^{s+1}} \le \mu(a)/2,
$$
for all $1 \le i \le \ell$. Thus we reached a contradiction.
\end{proof}

We now define $t:=\max(\mathcal{S})$. Theorem~\ref{thm:D:main} follows from the following two lemmas.

\begin{lemma}
\label{lemma:D:main_rightarrow}
Let $(f_n:\F_p^n \to \D)_{n \in \N}$ be a sequence of functions
such that $\Ex[f_n]=0$.
If $\lim_{n \rightarrow \infty} \|f_n\|_{U^{t+1}}=0$, then $\lim_{n \rightarrow \infty} \|f_n\|_{u(D_n)}=0$.
\end{lemma}

\begin{lemma}
\label{lemma:D:main_leftarrow}
Let $(f_n:\F_p^n \to \D)_{n \in \N}$ be a sequence of functions
such that $\Ex[f_n]=0$.
If $\lim_{n \rightarrow \infty} \|f_n\|_{u(D_n)}=0$, then $\lim_{n \rightarrow \infty} \|f_n\|_{U^{t+1}}=0$.
\end{lemma}

We prove Lemma~\ref{lemma:D:main_rightarrow} in Subsection~\ref{subsec:proof:lemma:D:main_rightarrow} and Lemma~\ref{lemma:D:main_leftarrow} in Subsection~\ref{subsec:proof:lemma:D:main_leftarrow}.

\subsection{Proof of Lemma~\ref{lemma:D:main_rightarrow}}
\label{subsec:proof:lemma:D:main_rightarrow}

Suppose that $\lim_{n \rightarrow \infty} \|f_n\|_{U^{t+1}} =0$, but
$$c:=\limsup_{n \rightarrow \infty} \|f_n\|_{u(D_n)} > 0.$$
Since $\mathcal{D}$ is consistent, we can replace the $\limsup$ by an actual limit. Assume that $c=\lim_{n \rightarrow \infty} \|f_n\|_{u(D_n)}$ and set $\eps := c/2$. Since $\mathcal{D}$ is strongly correlation testable with
true complexity $d$ and Cauchy-Schwarz complexity $s<p$, there exist $\delta \in (0,\eps)$, $\eps'>0$, and a family of homogeneous systems of linear forms $\mathcal{L}_1,\ldots,\mathcal{L}_\ell$ of Cauchy-Schwarz complexity at most $s$ and true complexity at most $d$ along with conjugations $\alpha_1,\ldots,\alpha_\ell$ such that
\begin{equation}
\label{eq:D:main_rightarrow:testingSeparation}
\|(t_{\mathcal{L}_1,\alpha_1}(f),\ldots, t_{\mathcal{L}_{\ell},\alpha_{\ell}}(f)) - (t_{\mathcal{L}_1,\alpha_1}(g),\ldots, t_{\mathcal{L}_{\ell},\alpha_{\ell}}(g))\|_\infty  \ge \eps',
\end{equation}
for every $f,g:\F_p^n \rightarrow \D$ (with $n>n_0$) satisfying $\|f\|_{u(D_n)} \ge \eps$ and $\|g\|_{u(D_n)} \le \delta$.

Let $r:\N \to \N$ be a growth function to be defined later. We apply Theorem~\ref{thm:decompose_multiple_func} and Theorem~\ref{thm:avg-approx-funcs} to deduce that there exists a polynomial factor $\mathcal{B}_n$ of degree $s$, complexity $C_n \le C_{\max}(s,\eps,r(\cdot))$ and rank at least $r(C_n)$, such that for $h_n:=\E(f_n|\mathcal{B}_n)$ we have
\begin{equation}
\label{eq:D:main_rightarrow:approx_f_avg}
\left| t_{\mathcal{L}_i,\alpha_i}(f_n) - t_{\mathcal{L}_i,\alpha_i}(h_n) \right| \le \eps'/2,
\end{equation}
for all $1 \le i \le \ell$.
Equations~\eqref{eq:D:main_rightarrow:testingSeparation} and~\eqref{eq:D:main_rightarrow:approx_f_avg} imply that for large enough $n$ we have  $\|h_n\|_{u(D_n)} > \delta$. So, for large enough $n$, there exists $g_n \in D_n$ such that
\begin{equation}
\label{eq:D:main_rightarrow:lowerL4gh}
|\ip{h_n,g_n}| > \delta.
\end{equation}
Let $\mathcal{B}_n$ be defined by polynomials $Q_{n,1},\ldots,Q_{n,C_n}$ and define $Q_{n,\gamma}:=\sum_{i=1}^{C_n} \gamma(i) Q_{n,i}(x)$ for every $\gamma \in \F_p^{C_n}$. By choosing the growth function $r(\cdot)$ large enough, we have by Theorem~\ref{thm:regularity} that for all $\gamma \ne \gamma'$,
\begin{equation}
\label{eq:D:main_rightarrow:near_orthogonality}
|\bias(Q_{n,\gamma}-Q_{n,\gamma'})| \le p^{-2C_n} \delta/100.
\end{equation}
As $h_n$ is $\mathcal{B}_n$-measurable, we can express it as $h_n(x) = F_n(Q_{n,1}(x),\ldots,Q_{n,C_n}(x))$ for some $F_n:\F_p^{C_n} \rightarrow \D$. Consider the Fourier decomposition of $F_n$,
$$
F_n(z_1,\ldots,z_{C_n}) = \sum_{\gamma \in \F_p^{C_n}} \widehat{F_n}(\gamma) \exp\left(\sum_{i=1}^{C_n} \gamma(i) z_i\right).
$$
We thus have
$$h_n(x) = \sum_{\gamma \in \F_p^{C_n}} \widehat{F_n}(\gamma) \exp(Q_{n,\gamma}(x)),$$
where $|\widehat{F_n}(\gamma)| \le 1$. Define $W_n:=\{\gamma \in \F_p^{C_n}: \deg(Q_{n,\gamma}) \le t\}$. We now show that the assumption $\lim_{n \to \infty} \|f_n\|_{U^{t+1}}=0$ implies that by taking $n$ large enough, we can make $|\widehat{F_n}(\gamma)|$ arbitrarily small for all $\gamma \in W_n$.
\begin{claim}
\label{claim:D:h_no_corr_deg_k-1}
For large enough $n$, we have $|\widehat{F_n}(\gamma)| \le p^{-C_n} \delta/10$ for all $\gamma \in W_n$.
\end{claim}

\begin{proof}
By expanding $\ip{h_n,\exp(Q_{n,\gamma})}$ we have
$$
\widehat{F_n}(\gamma) = \ip{h_n,\exp(Q_{n,\gamma})} - \sum_{\gamma' \ne \gamma} \widehat{F_n}(\gamma') \cdot \bias(Q_{n,\gamma'}-Q_{n,\gamma}).
$$
We bound each term individually. As $\exp(Q_{n,\gamma})$ is $\mathcal{B}_n$-measurable and $h_n=\E(f_n|\mathcal{B}_n)$, we have that
$$
\ip{h_n,\exp(Q_{n,\gamma})} = \ip{f_n,\exp(Q_{n,\gamma})},
$$
and since $\deg(Q_{n,\gamma}) \le t$ we have
$$
|\ip{h_n,\exp(Q_{n,\gamma})}| \le \|f_n\|_{U^{t+1}} \le p^{-C_n} \delta/100,
$$
for large enough $n$. By~\eqref{eq:D:main_rightarrow:near_orthogonality} and the bound $\|\widehat{F_n}\|_{\infty} \le 1$, we conclude that
$$
|\widehat{F_n}(\gamma)| \le p^{-C_n} \delta/10,
$$
for all $\gamma \in W_n$.
\end{proof}

Now, Claim~\ref{claim:D:h_no_corr_deg_k-1} implies that
$$
\left|\sum_{\gamma \in W_n} \widehat{F_n}(\gamma) \ip{\exp(Q_{n,\gamma}), g_n}\right| \le \delta/2.
$$
However, since $|\ip{h_n,g_n}| \ge \delta$ we must have that there exists $\gamma^* \not\in W_n$ such that
\begin{equation}
\label{eq:D:main_rightarrow:lowerBoundCorLargeDeg}
|\ip{\exp(Q_{n,\gamma^*}), g_n}| \ge p^{-C_n} \delta/2.
\end{equation}
We now show a contradiction to the assumption that $t=\max(\mathcal{S})$. Set $P_n := Q_{n,\gamma^*}$. As $Q_{n,\gamma} \notin W_n$,  by~(\ref{eq:D:main_rightarrow:lowerBoundCorLargeDeg}) we have that
\begin{itemize}
\item $t+1 \le \deg(P_n) \le s$;
\item $\|\exp(P_n)\|_{u(D_n)} \ge p^{-C_n} \delta/2$;
\item $\rank(P_n) \ge \rank(\mathcal{B}_n) \ge r(C_n)$.
\end{itemize}
Let $n_1<n_2<\ldots$ be an infinite sequence such that $t':=\deg(P_{n_i}) \ge t+1$ and $C_n=C$. Since the family $\mathcal{D}$ is consistent, we may assume (by refining the sequence) that $\deg(P_n)=t'$ and $\rank(P_n) \ge r(C)$ for all $n \in \N$. As $r(\cdot)$ is an arbitrary growth function, we must have $t' \in \mathcal{S}$ and the lemma follows.


\subsection{Proof of Lemma~\ref{lemma:D:main_leftarrow}}
\label{subsec:proof:lemma:D:main_leftarrow}
Let $(f_n:\F_p^n \to \D)_{n \in \N}$ be a sequence of functions such that
$\lim_{n \to \infty} \|f_n\|_{u(D_n)}=0$, and assume to the contrary that
$$
\limsup_{n \to \infty} \|f_n\|_{U^{t+1}} >0.
$$
Let $t_0 \ge 1$ be the smallest positive integer for which  
$$
\zeta:=\limsup_{n \to \infty} \|f_n\|_{U^{t_0+1}} >0.
$$ 
Since $\mathcal{D}$ is consistent we may replace the $\limsup$ by an actual limit by refining the sequence. So we assume that $\|f_n\|_{U^{t_0+1}} \ge \zeta$ for all large enough $n>n_0$. By Theorem~\ref{thm:inverse} this
implies (since $t_0 \le t \le s<p$) that there exist polynomials $Q_n$ of degree at most $t_0$ such that
$$
|\ip{f_n,\exp(Q_n)}| \ge \frac{1}{b},
$$
for some integer $b=b(\zeta)>0$.  Since $t_0$ is minimal, we can assume after possibly refining the sequence that $\deg(Q_n)=t_0$ for every $n$.  We first show that the polynomials will have arbitrarily large rank.

\begin{claim}
\label{claim:D:regularize_Qn}
We have $\lim_{n \to \infty} \rank(Q_{n}) = \infty$.
\end{claim}

\begin{proof}
 Otherwise, there exist an integer $c$, and an infinite sequence  $n_1<n_2<\ldots$ such that $\rank(Q_{n_i})<c$ for every $i \in \mathbb{N}$.  Consider a particular $n:=n_i$ and let $Q:=Q_{n_i}$. By our assumption we can express  $Q$  as a function of $c$ polynomials $Q'_{1},\ldots,Q'_{c}$ of degrees at most $t_0-1$. Assume that $Q(x) = F(Q'_{1},\ldots,Q'_{c})$. Let $Q'_{\gamma}=\sum_{i=1}^{c} \gamma(i) Q'_{i}(x)$ for $\gamma \in \F_p^{c}$. By the Fourier decomposition of $F$, we have
$$
\exp(Q(x)) = \sum_{\gamma \in \F_p^{c}} \widehat{F}(\gamma) \exp\left(Q'_{\gamma}(x)\right).
$$
Since $\|\widehat{F}\|_{\infty} \le 1$, we have that there exists $\gamma$ with
$$
|\ip{f_n,\exp(Q'_{\gamma})}| \ge \frac{p^{-c}}{b}.
$$
Since we assumed $\Ex[f_n]=0$, we cannot have that $Q'_{\gamma}$ is a constant, thus we have $1 \le \deg(Q'_{\gamma}) < t_0$. This shows that 
$\|f_n\|_{U^{\deg(Q'_{\gamma})+1}} \ge \frac{p^{-c}}{b}$
which contradicts the minimality of $t_0$. 
\end{proof} 

So far we have established that there exist $1 \le t_0 \le t$,  functions $f_n:\F_p^n \to \D$
with $\lim_{n \to \infty} \|f_n\|_{u(D_n)}=0$, and polynomials $Q_n$ of degree exactly $t_0$, such that $\lim_{n \to \infty} \rank(Q_{n}) = \infty$, and for every $n>n_0$, we have $|\ip{f_n, \exp(Q_n)}| \ge \frac{1}{b}$.
We first derive a contradiction when $t_0 \in \mathcal{S}$.
\begin{lemma}
\label{lemma:D:t_in_S}
If $t_0 \in \mathcal{S}$, then $\limsup_{n \rightarrow \infty} \|f_n\|_{u(D_n)} > 0$.
\end{lemma}

\begin{proof}
In the proof, we think of $p$ and $s$ as constants and do not explicitly mention dependencies on them. Given the value of $\delta$, let $\widetilde{r}_b:\N \to \N$ be a growth function to be determined later. Since $t_0 \in \mathcal{S}$, there exist polynomials $P_n$ of degree exactly $t_0$, functions $g_n \in D_n$, and $n_0,a \in \N$, such that for every $n>n_0$,
\begin{itemize}
\item $|\ip{\exp(P_n),g_n}| \ge 1/a$;
\item $\rank(P_n) \ge \widetilde{r}_b(a)$.
\end{itemize}
Set $\eps:=\tfrac{1}{1000 a b}$. Since $\mathcal{D}$ is correlation testable with true complexity $d$ and Cauchy-Schwarz complexity $s<p$, there exist $\delta \in (0,\eps)$, $\eta>0$, and a family of homogeneous systems of linear forms $\mathcal{L}_1,\ldots,\mathcal{L}_{\ell}$ of true complexity at most $d$ and Cauchy-Schwarz complexity at most $s$, and conjugations $\alpha_1,\ldots,\alpha_{\ell}$, such that
$$
\|(t_{\mathcal{L}_1,\alpha_1}(f),\ldots, t_{\mathcal{L}_{\ell},\alpha_{\ell}}(f)) - (t_{\mathcal{L}_1,\alpha_1}(f'),\ldots, t_{\mathcal{L}_{\ell},\alpha_{\ell}}(f'))\|_\infty \ge 2\eta,
$$
for every $f,f':\F_p^n \rightarrow \D$ (with $n>n_0$) satisfying $\|f'\|_{u(D_n)} \ge \eps$ and $\|f\|_{u(D_n)} \le \delta$. We will prove the lemma by constructing a new sequence of functions $(f'_n:\F_p^n \to \D)_{n \in \N}$ such that for large enough $n$, we will have
\begin{itemize}
\item $|\ip{f'_n,g_n}| \ge \eps$ and hence $\|f'_n\|_{u(D_n)} \ge \eps$;
\item $|t_{\mathcal{L}_i,\alpha_i}(f_n)-t_{\mathcal{L}_i,\alpha_i}(f'_n)| \le \eta$, for all $1 \le i \le k$.
\end{itemize}
This will conclude the lemma as it will show that $\|\exp(f_n)\|_{u(D_n)} \ge \delta$ for large enough $n$.

Let $r_1:\N \to \N$ be a growth function to be determined later (whose choice depends on the values of $a,b$). Since $\rank(Q_n) \to \infty$, 
By Theorem~\ref{thm:strong_independence}, for sufficiently large $n$, there exists a polynomial factor $\mathcal{B}_n$ of degree $s$, complexity $C_n \le C_{\max}(\eta,a,b,r_1(\cdot),1)$ and rank at least $r_1(C_n)$, defined by polynomials $Q'_{n,1},\ldots,Q'_{n,C_n}$ with $Q'_{n,1}=Q_n$ such that the following holds. For $h_n:=\E(f_n|\mathcal{B}_n)$ we have
\begin{equation}
\label{eq:D:t_in_S:f_h_close_L}
\left|t_{\mathcal{L}_i,\alpha_i}(f_n) - t_{\mathcal{L}_i,\alpha_i}(h_n)\right| \le \eta/10,
\end{equation}
for all $1 \le i \le k$, and also
\begin{equation}
\label{eq:D:t_in_S:f_h_close_gowers}
\|f_n - h_n\|_{U^{s+1}} \le 1/2b.
\end{equation}

Recall that $Q_n$ is a polynomial of degree $t_0 \le t \le s$ such that $|\ip{f_n,Q_n}| \ge 1/b$, for $n \ge n_0$. By the choice of~\eqref{eq:D:t_in_S:f_h_close_gowers} we have that $h_n$ is also correlated to $Q_n$, as
\begin{equation}
\label{eq:D:t_in_S:h_correlated_Q}
|\ip{h_n,\exp(Q_n)}| \ge |\ip{f_n,\exp(Q_n)}| - |\ip{f_n-h_n,\exp(Q_n)}| \ge \frac{1}{b} - \|f_n-h_n\|_{U^{s+1}} \ge 1/2b.
\end{equation}

Define $Q'_{n,\gamma} = \sum_{i=1}^{C_n} \gamma(i) Q'_{n,i}$ for $\gamma \in \F_p^{C_n}$. The function $h_n=\E(f_n|\mathcal{B}_n)$ is $\mathcal{B}_n$-measurable, hence $h_n=F_n(Q'_{n,1},\ldots,Q'_{n,C_n})$ for some function $F_n:\F_p^{C_n} \to \D$, and we have
$$
h_n(x) = \sum_{\gamma \in \F_p^{C_n}} \widehat{F_n}(\gamma) \exp(Q'_{n,\gamma}(x))
$$
where $\|\widehat{F_n}\|_{\infty} \le 1$. Thus we have
\begin{equation}
\label{eq:D:t_in_S:correlation_inequality}
1/2b \le \bigg|\ip{h_n,\exp(Q_n)}\bigg| \le \sum_{\gamma \in \F_p^{C_n}} |\widehat{F_n}(\gamma)| \cdot |\bias(Q'_{n,\gamma}-Q_n)|.
\end{equation}

We now show that when $r_1(\cdot)$ is chosen large enough (as a function of $a,b$), then almost all the contribution to the correlation in~\eqref{eq:D:t_in_S:correlation_inequality} comes from the single term $\gamma=\e_1=(1,0,\ldots,0)$. Let $r_1(C)$
be chosen large enough, such that if $R$ is a polynomial on $\F_p^n$ of degree at most $s$ and rank at least $r_1(C)$, then for large enough $n$ we have
\begin{equation}
\label{eq:D:t_in_S:bias_upper_bound}
\left| \bias(R) \right| \le \eps \cdot p^{-C}.
\end{equation}
Such a choice is guaranteed by Theorem~\ref{thm:regularity}. Assuming such a choice for $r_1(\cdot)$, we have 
$$
\sum_{\gamma \in \F_p^{C_n}, \gamma \neq \e_1} |\widehat{F_n}(\gamma)| |\bias(Q'_{n,\gamma}-Q_n)| \le \eps \le \frac{1}{10b},
$$ 
Thus \eqref{eq:D:t_in_S:correlation_inequality} implies that  $|\widehat{F_n}(\e_1)| \ge 1/4b$. 

Summarizing the discussion above, we have established the following properties:
\begin{enumerate}
\item $Q'_{n,1} = Q_n$;
\item $|\widehat{F_n}(\e_1)| \ge \tfrac{1}{4b}$;
\item $|\bias(Q'_{n,\gamma}-Q_n)| \le \eps \cdot p^{-C_n}$ for all $\gamma \ne \e_1$;
\end{enumerate}

We now repeat the same process for $g_n$. Let $\lambda = \eps p^{-C_{\max}}$. There exists a polynomial factor $\mathcal{B'}_n$ of degree $s$ and complexity $C'_n \le C_{max}$ defined by polynomials $\{P'_{n,i}: 1 \le i \le C'_n\}$, such that for $g'_n:=\E(g_n|\mathcal{B'}_n)$ we have
$$
\|g_n - g'_n\|_{U^{s+1}} \le \lambda,
$$
and also that $g'_n=G_n(P'_{n,1},\ldots,P'_{n,C'_n})$ where
\begin{enumerate}
\item $P'_{n,1} = P_n$; (To achieve this, we choose $\tilde{r}_b(a)$ to be sufficiently large so that $P_n$ is never decomposed). 
\item $|\widehat{G_n}(\e_1)| \ge \tfrac{1}{4a}$;
\item $|\bias(P'_{n,\gamma}-P_n)| \le \eps \cdot p^{-C'_n}$ for all $\gamma \ne \e_1$;
\end{enumerate}

We now define a new polynomial factor $\mathcal{\widetilde{B}}_n$ as follows. Let $R_{n,2},\ldots,R_{n,C_n}$ be random  polynomials in $\F_p^n$ chosen such that $\deg(R_{n,i})=\deg(P'_{n,i})$. The same argument as in Remark~\ref{rankrandom} shows that  for every rank bound $r^*$,
for large enough $n$ the following holds with high probability. Let $P^{1}$ be some (possibly zero) linear combination of $P_{n},P'_{n,1},\ldots,P'_{n,C'_n}$ and let $R^1$ be a nonzero linear combination of $R_{n,2},\ldots,R_{n,C_n}$.
Then as long as $\deg(R^1) \ge \deg(P^1)$, we have that $\rank(P^1+R^1) \ge r^*$. 

We define $\mathcal{\widetilde{B}}_{n}=\{P_n,R_{n,2},\ldots,R_{n,C_n}\}$. Define $R_{n,\gamma}=\gamma(1) P_n(x) + \sum_{i=2}^{C_n} \gamma(i) R_{n,i}(x)$ for $\gamma \in \F_p^{C_n}$. We define the new sequence of functions as
\begin{equation}
\label{eq:D:t_in_S:def_f'}
f'_n(x) = F_n(P_n,R_{n,2},\ldots,R_{n,C_n}).
\end{equation}

We conclude the proof by showing that $|t_{\mathcal{L}_i,\alpha_i}(f'_n)-t_{\mathcal{L}_i,\alpha_i}(f_n)| \le \eta$ for all $1 \le i \le \ell$, but that $|\ip{f'_n,g_n}| \ge \eps$.

\begin{claim}
\label{claim:D:f'_has_same_avg_as_f}
$|t_{\mathcal{L}_i,\alpha_i}(f'_n)-t_{\mathcal{L}_i,\alpha_i}(f_n)| \le \eta$ for all $1 \le i \le \ell$.
\end{claim}

\begin{proof}
Note that $\deg(R_{n,i})=\deg(P_{n,i})$ for all $1 \le i \le \ell$. Since all the linear forms are homogeneous, we can apply Proposition~\ref{prop:invariance}. By the proposition,
there exists a bound $r'(\eta)$ such that if $\rank(\mathcal{B}_n),\rank(\mathcal{\widetilde{B}}_n) \ge r'$ then the claim follows. To ensure this, we require for $\mathcal{B}_n$ that $r_1(C) \ge r'$ for all $C \in \N$; and for $\mathcal{\widetilde{B}}_n$ that $r^* \ge r'$ and that $\widetilde{r}_b(a) \ge r'$ (note that it is crucial to allow $\widetilde{r}$ to depend on both $a,b$, since $\eta$ depends on both $a,b$).
\end{proof}

\begin{claim}
$|\ip{f'_n,g_n}| \ge \eps$.
\end{claim}

\begin{proof}
We first claim, it suffices to prove $|\ip{f'_n,g'_n}| \ge 2\epsilon$. Indeed
$$
|\ip{f'_n,g_n-g'_n}| \le \sum_{\gamma \in \F_p^{C_n}} |\widehat{F_n}(\gamma)| |\ip{\exp(R_{n,\gamma}),g_n-g'_n}|
\le p^{C_n} \|g_n - g'_n\|_{U^{s+1}} \le \eps,
$$
where we used the fact that $\deg(R_{n,\gamma}) \le s$. Now note that
$$
\ip{f'_n,g'_n} = \sum_{\gamma \in \F_p^{C_n},\gamma' \in \F_p^{C'_n}} \widehat{F_n}(\gamma) \overline{\widehat{G_n}(\gamma')}
\bias(R_{n,\gamma}-P'_{n,\gamma'}).
$$
Consider first the term $\gamma=\gamma'=\e_1$. We have $R_{n,\e_1}=P'_{n,\e_1}=P_n$ and by our construction,
$$
\left|\widehat{F_n}(\e_1) \overline{\widehat{G_n}(\e_1)} \right| \ge \frac{1}{4b} \cdot \frac{1}{4a} \ge 3 \eps.
$$
We now bound all other terms $(\gamma,\gamma') \ne (\e_1,\e_1)$. By choosing $r_1(\cdot),r^*$ large enough, we can bound
$$
|\bias(R_{n,\gamma}-P'_{n,\gamma'})| \le \eps \cdot p^{-(C_n+C'_n)}.
$$
Putting all these together we conclude that $|\ip{f'_n,g_n}| \ge \eps$ as claimed.
\end{proof}

This concludes the proof of Lemma~\ref{lemma:D:t_in_S}.
\end{proof}

Lemma~\ref{lemma:D:t_in_S} shows that if $t_0 \in \mathcal{S}$, then we are done. We will next show that
$\mathcal{S}=\{1,\ldots,t\}$ which will conclude the proof of Lemma~\ref{lemma:D:main_leftarrow}. Consider any $t'<t$ and systems of linear forms $\mathcal{L}_1,\ldots,\mathcal{L}_{\ell}$. For a function $f:\F_p^n \to \mathbb{R}$ we shorthand $t_{\mathcal{L}}(f)=t_{\mathcal{L},\alpha}(f)$ for all conjugates $\alpha$ since $f=\overline{f}$. In Lemma~\ref{lem:D:identicalProfile} below we shall show that it is possible to construct two different families of functions $f_n,f'_n: \F_p^n \rightarrow [-1,1]$ such that

\begin{enumerate}
\item[(i)] $f_n$ and $f'_n$ cannot be distinguished by averages $t_{\mathcal{L}_i}$;
\item[(ii)] $f_n$ has correlation with polynomials of degree exactly $t$ and of arbitrarily high rank;
\item[(iii)] $f'_n$ is a linear combination of a bounded number of exponentials of polynomials of degree exactly $t'$ and arbitrarily high rank.
\end{enumerate}

We note that the combination of (i), (ii), (iii) implies that $t' \in \mathcal{S}$:  Since $t \in \mathcal{S}$, Condition~(ii) implies by Lemma~\ref{lemma:D:t_in_S} that $\limsup_{n \rightarrow \infty} \|f_n\|_{u(D_n)} > 0$. By Condition (i) this implies that also
$\limsup_{n \rightarrow \infty} \|f'_n\|_{u(D_n)} > 0$; so there exist $g_n \in D_n$ such that $\limsup_{n \rightarrow \infty} |\ip{f'_n,g_n}|=\delta>0$. By Condition (iii) we can express $f'_n(x) = \sum_{i=1}^C a_i \exp(Q_{n,i}(x))$ where $C$ is a uniform bound, $|a_i| \le 1$ and $Q_{n,i}$ are polynomials of degree $t'$ and arbitrarily high rank. Hence, we must have
for infinitely many $n$ that $|\ip{\exp(Q_{n,i}),g_n}| \ge \delta p^{-C}$ for some $i$. Since $\mathcal{D}$ is consistent, we can extend this to all large enough $n$ and complete the proof.

It only remains to prove the following lemma whose proof uses the main result of Section~\ref{sec:interior}.

\begin{lemma}
\label{lem:D:identicalProfile}
Let $\mathcal{L}_1,\ldots,\mathcal{L}_{\ell}$ be systems of linear forms. Let $t'<t$. There exist functions $f_n,f'_n: \F_p^n \rightarrow [-1,1]$ and a constant $C \in \mathbb{N}$ such that the following holds:
\begin{enumerate}
\item[(i)] For every $i \in [\ell]$,
$$
t_{\mathcal{L}_i}(f_n) =  t_{\mathcal{L}_i}(f'_n).
$$

\item[(ii)] There exist polynomials $P_n:\F_p^n \rightarrow \F_p$ satisfying $\deg(P_n)=t$ and $\lim_{n \rightarrow \infty} \rank(P_n)=\infty$, and
$$\liminf_{n \rightarrow \infty} |\ip{ \exp(P_n), f_{n} }| > 0.$$

\item[(iii)] $f'_{n}$ is a linear combination of exponentials of $C$ high rank polynomials of degree exactly $t'$. That is, there exist polynomials $Q_{n,1},\ldots,Q_{n,C}:\F_p^n \rightarrow \F_p$ satisfying $\deg(Q_{n,i})=t'$ and $\lim_{n \rightarrow \infty} \rank(Q_{n,i})=\infty$ for all $1 \le i \le C$, and
$$f'_n(x) = \sum_{i=1}^C a_i \exp(Q_{n,i}(x)),$$
where $|a_i| \le 1$.
\end{enumerate}
\end{lemma}

\begin{proof}
We first note that it is sufficient to prove the lemma for systems of linear forms which are non-isomorphic and connected, since we can decompose each system to its connected components and remove isomorphic copies.

First let us introduce some notations. For positive integers $m>n$, let $\pi_{m \to n}$ denote the natural projection from $\F_p^m$ to $\F_p^n$ defined as
$$
\pi_{m \rightarrow n}:(x_1,\ldots,x_m) \mapsto (x_1,\ldots,x_n).
$$
For functions $f:\F_p^m \rightarrow \C$ and $g:\F_p^n \rightarrow \C$, let $f \otimes g:\F^{m+n} \rightarrow \C$ denote the function
$$
(f \otimes g)(x_1,\ldots,x_{m+n}) = f(x_1,\ldots,x_m) g(x_{m+1},\ldots,x_{m+n}),
$$
and note that for every system of linear forms $\mathcal{L}$, we have $$t_\mathcal{L}(f \otimes g)= t_\mathcal{L}(f) t_\mathcal{L}(g).$$

By Theorem~\ref{thm:InteriorLinearForms} that will be stated and proved later in Section~\ref{sec:interior}, there exist a constant $N \in \mathbb{N}$, an $\eps>0$, and a function $F:\F_p^N \rightarrow [0,1]$ such that
\begin{equation}
\label{eq:D:interior}
\left\{ z \in \mathbb{R}^{\ell} \ | \ \|z - (t_{\mathcal{L}_1}(F),\ldots,t_{\mathcal{L}_{\ell}}(F))\|_\infty \le \eps \right\} \subseteq  \left\{(t_{\mathcal{L}_1}(f),\ldots,t_{\mathcal{L}_{\ell}}(f))\ |\  f:\F_p^N \rightarrow [0,1]\right\}.
\end{equation}
Consider two sequences of polynomials $P_n,Q_n:\F_p^n \rightarrow \F_p$ such that $\deg(P_n)=t$ and $\deg(Q_n)=t'$, and
$$\lim_{n \rightarrow \infty} \rank(P_n)= \lim_{n \rightarrow \infty} \rank(Q_n) = \infty.$$
Define $g_n: \F_p^n \rightarrow \{-1/p,1-1/p\}$ as $g_n(x)=1-1/p$ if and only if $Q_n(x)=0$. Note that
$$g_n(x) = \sum_{\alpha \in \F_p \setminus \{0\}} \frac{1}{p} \exp(\alpha Q_n(x)).$$
Define $h_n: \F_p^n \rightarrow \{0,1\}$ as $h_n(x)=1$ if and only if $P_n(x)=0$.

Let $\delta>0$ be sufficiently small so that for every $m>N$, and every $i \in [k]$,
$$\left| t_{\mathcal{L}_i}(F) -  t_{\mathcal{L}_i}(\delta h_m + (1-\delta) F \circ \pi_{m \rightarrow N})\right| \le \eps/2.$$
Then by (\ref{eq:D:interior}), for every $m>N$, there exists a function $G_m:\F_p^N \rightarrow  [0,1]$ such that for every $i \in [k]$,
\begin{equation}
\label{eq:D:GnTheSame}
t_{\mathcal{L}_i}(G_m) =  t_{\mathcal{L}_i}(\delta h_m + (1-\delta) F \circ \pi_{m \rightarrow N}).
\end{equation}
For $n>2N$, let $m=\lfloor n/2 \rfloor$ and define $f_{n}, f'_{n}:\F_p^{n} \rightarrow [-1,1]$ as
\begin{eqnarray*}
f_{n} &:=& g_{n-m}\otimes  (\delta h_m + (1-\delta) F \circ \pi_{m \rightarrow N}),\\
f'_{n}  &:=& g_{n-m}\otimes  (G_{m} \circ \pi_{m \to N}).
\end{eqnarray*}
By (\ref{eq:D:GnTheSame})  for every $i \in [k]$ and every $n>2N$,
$$t_{\mathcal{L}_i}(f_n) = t_{\mathcal{L}_i}(g_{n-m}) t_{\mathcal{L}_i}(\delta h_m + (1-\delta)F \circ \pi_{m \rightarrow N}) = t_{\mathcal{L}_i}(g_{n-m}) t_{\mathcal{L}_i}(G_m \circ \pi_{m \to N})
= t_{\mathcal{L}_i}(f'_n),$$
which establishes (i). To establish (ii), let $R_n := Q_{n-m} \otimes P_m$. Note that $R_n$ is a polynomial of degree $t$ and $\lim_{n \to \infty}\rank(R_n)=\infty$. We have
$$
\ip{f_n,\exp(R_n)} = \ip{g_{n-m},\exp(Q_{n-m})} \cdot \ip{\delta h_m + (1-\delta) F \otimes \pi_{m \to N},\exp(P_m)} .
$$
We now lower bound the terms. By the definition of $g_{n-m}$, we have
$$
|\ip{g_{n-m},\exp(Q_{n-m})}| = \left|\frac{1}{p} + \frac{1}{p} \sum_{\alpha \in \F_p \setminus \{0,1\}} \bias((\alpha-1) Q_{n-m})\right| \ge \frac{1}{2p},
$$
for large enough $n-m$ since $\rank(Q_{n-m}) \to \infty$. The function $F \otimes \pi_{m \to N}$ depends only on the first $N$ variables; hence we have $\lim_{m \to \infty} \ip{F \otimes \pi_{m \to N},\exp(P_m)}=0$ since $\rank(P_m) \to \infty$ by Theorem~\ref{thm:regularity}. Finally, by the definition of $h_m$ we have
$$
\ip{h_m,\exp(P_m)}=\Pr_X[P_m(X)=0] \ge \frac{1}{2p}
$$
for large enough $m$, since $\rank(P_m) \to \infty$. We thus conclude that
$$
|\ip{f_n,\exp(R_n)}| \ge \frac{1}{4p^2}
$$
for large enough $n$, which establishes (ii). To conclude the proof we establish (iii). Let
$$
G_m(x) = \sum_{\gamma \in \F_p^N} \widehat{G}(\gamma) \exp\left(\sum_{i=1}^N \gamma(i) x(i)\right)
$$
where $|\widehat{G}(\gamma)| \le 1$. We thus have
\begin{eqnarray*}
f'_n(x)&=&g_{n-m}(x(1),\ldots,x(n-m)) G_m(x(n-m+1),\ldots,x(n-m+N)) \\
&=&
\sum_{\gamma \in \F_p^N, \alpha \in \F_p \setminus \{0\}} \frac{1}{p} \widehat{G}(\gamma) \exp\left(\alpha Q_{n-m}(x(1),\ldots,x(n-m)) + \sum_{i=1}^N \gamma(i) x(n-m+i)\right).
\end{eqnarray*}
Hence, we can express $f'_n$ as the linear combination of $C=(p-1)p^N$ exponentials of polynomials of degree exactly $t'$; and as $N$ is fixed and $n-m \to \infty$, their rank is unbounded. This established (iii) and concludes the proof.
\end{proof}

\section{Characterization of correlation testable properties}
\label{sec:F:testing_correlation}

Consider a family ${\mathcal D}:=\{D_n\}_{n \in \mathbb{N}}$ where $D_n$ is a set of functions from $\F_p^n$ to $\F_p$. Given a function $\ff_n:\F_p^n \rightarrow \F_p$, we want to probabilistically determine whether $\ff_n$ has correlation with $D_n$. That is, we want to determine if $\|\exp(\ff_n)\|_{u(D_n)}$ is non-negligible or not, where we are only allowed to read the value of $\ff_n$ on a few points.  This was made into a precise definition in Definition~\ref{def:corTestable}.

We study proper dual families $\mathcal{D}$. We recall that a family $\mathcal{D}$ is {\em proper dual} if the following conditions (first introduced in Section~\ref{sec:introCorTesting})   hold:
\begin{itemize}
\item {\bf A1: Consistency} For positive integers $m>n$ and $\gg \in D_n$, the function $\hh: \F_p^m \rightarrow \F_p$ defined as $\hh(x_1,\ldots,x_m)=\gg(x_1,\ldots,x_n)$ belongs to $D_m$.

\item {\bf A2: Affine invariance} For every positive integer $n$, if $\gg \in D_n$, then for every $A \in \Aff(n,\F_p)$, we have $A\gg \in D_n$.

\item {\bf A3: Sparsity} For every $\eps>0$ and large enough $n$, we have $|D_n| \le p^{\eps p^n}$.
\end{itemize} 

\subsection{Correlation testing by averages over linear forms}

We first show that if $\mathcal{D}$ is a proper dual family which is correlation testable using $q$ queries, then it is in fact also testable using averages of linear forms. When arguing about functions to $\F_p$, one may allow more general types of averages. Let $\mathcal{L}=\{L_1,\ldots,L_m\}$ be a system of linear forms in $k$ variables. Let $\beta \in \F_p^m$ be a vector of coefficients. Recall that for a function $\ff:\F_p^n \to \F_p$, we define the average
$$
t^*_{\mathcal{L},\beta}(\ff) = \Ex_{\X \in (\F_p^n)^k} \left[\exp \left(\sum_{i=1}^m \beta(i) \ff(L_i(\X)) \right)\right].
$$
We note that for functions $\ff:\F_p^n \to \F_p$, these averages generalize the previous averages $t_{\mathcal{L},\alpha}$ which were defined for bounded functions. Indeed, for $\alpha \in \{0,1\}^m$ let $\beta \in \{-1,1\}^m$ be defined as $\beta(i) = (-1)^{\alpha(i)}$, then
$$
t^*_{\mathcal{L},\beta}(\ff) = t_{\mathcal{L},\alpha}(\exp(\ff)).
$$

\begin{lemma}
\label{lemma:F:correlation_testable_by_linear_forms_with_few_queries}
Suppose that a proper dual family ${\mathcal D}=\{D_n\}$ is correlation testable with $q$ queries. Then for every $\eps>0$, there exists  $\delta \in (0,\eps)$, $n_0 \in \mathbb{N}$, and homogeneous systems of linear forms
$\mathcal{L}_1,\ldots,\mathcal{L}_{\ell}$ with $m_1,\ldots,m_{\ell}$ linear forms, accordingly, and corresponding coefficients $\beta_1 \in \F_p^{m_1},\ldots,\beta_{\ell} \in \F_p^{m_{\ell}}$, such that the closures of the following two sets are disjoint:
$$T_\eps := \{(t^*_{\mathcal{L}_1,\beta_1}(\ff),\ldots, t^*_{\mathcal{L}_{\ell},\beta_{\ell}}(\ff)) | n >n_0,
\ff:\F_p^n \rightarrow \F_p, \|\exp(\ff)\|_{u(D_n)} \ge \eps\},$$
and
$$S_\eps := \{(t^*_{\mathcal{L}_1,\beta_1}(\ff),\ldots, t^*_{\mathcal{L}_{\ell},\beta_{\ell}}(\ff)) | n>n_0, \ff:\F_p^n \rightarrow \F_p, \|\exp(\ff)\|_{u(D_n)} \le \delta\}.$$
Moreover, the systems $\mathcal{L}_1,\ldots,\mathcal{L}_{\ell}$ have Cauchy-Schwarz complexity at most $q-2$.
\end{lemma}

\begin{proof}
Since $D_n$ is a proper dual, by Condition $\bf A2$, for every $\ff_n:\F_p^n \to \F_p$ and every $A \in \Aff(n,\F_p)$, we have that
$\|\exp(A \ff_n)\|_{u(D_n)} = \|\exp(\ff_n)\|_{u(D_n)}$. Let $A \in \Aff(n,\F_p)$ be a uniform random invertible affine transformation. Then by the assumption that $\mathcal{D}$ is correlation testable, we have that
\begin{itemize}
\item If $\|\exp(\ff_n)\|_{u(D_n)}\ge \eps$, then $\Pr_{(X_1,\ldots,X_q) \sim \mu, A \in \Aff(n,\F_p) }[\Gamma(\ff_n(A X_1),\ldots,\ff_n(A X_q))=1] \ge \theta^+$.
\item If $\|\exp(\ff_n)\|_{u(D_n)}\le \delta$, then $\Pr_{(X_1,\ldots,X_q) \sim \mu,, A \in \Aff(n,\F_p)}[\Gamma(\ff_n(AX_1),\ldots,\ff_n(AX_q))=1] \le \theta^-$.
\end{itemize}
We establish the lemma by showing that if we set $n_0$ large enough, then the probability
$$
\Pr_{(X_1,\ldots,X_q) \sim \mu, A \in \Aff(n,\F_p)}[\Gamma(\ff(AX_1),\ldots,\ff(AX_q))=1]
$$
can be approximated with an arbitrarily small error by linear combinations of $t^*_{\mathcal{L}_1,\beta_1}(\ff),\ldots,t^*_{\mathcal{L}_{\ell},\beta_{\ell}}(\ff)$, where $\{(\mathcal{L}_i,\beta_i)\}_{1 \le i \le \ell}$ are all possible homogeneous systems of at most $q$ linear forms. Note that $\ell$ is a constant depending only on $p,q$, and that the Cauchy-Schwarz complexity of any homogeneous system of $t \le q$ linear forms is at most $t-2 \le q-2$.
We start by decomposing $\Gamma$ to its Fourier decomposition
$$
\Gamma(z(1),\ldots,z(q)) = \sum_{\gamma \in \F_p^q} \widehat{\Gamma}(\gamma) \exp\left(\sum_{i=1}^q
\gamma(i) \cdot z(i)\right).
$$
We thus have that
$$
\Pr_{(X_1,\ldots,X_q) \sim \mu,A \sim \Aff(n,\F_p)}[\Gamma(\ff(AX_1),\ldots,\ff(AX_q))=1] =
\sum_{\gamma \in \F_p^q} \widehat{\Gamma}(\gamma) \Ex \left[ \exp\left(\sum_{i=1}^q \gamma(i) \cdot \ff(A X_i)\right) \right],
$$
where the expectation is taken over $(X_1,\ldots,X_q) \sim \mu$ and $A \sim \Aff(n,\F_p)$. Thus, it is enough to
show that each term $\Ex \left[\exp(\sum_{i=1}^q \beta(i) \cdot \ff(A X_i))\right]$ can be approximated by linear combinations of $\{t^*_{\mathcal{L}_i,\beta_i}(\ff)\}_{1 \le i \le k}$.

Fix $(x_1,\ldots,x_q) \in (\F_p^n)^q$. Suppose that the rank of ${\rm span}\{x_1,\ldots,x_q\}$ over $\F_p$ is $r$. Let $y_1,\ldots,y_r \in \F_p^n$ form a basis for ${\rm span} \{x_1,\ldots,x_q\}$, so that $x_i = \sum_{j=1}^r \lambda_{i,j} y_j$, for every $1 \le i \le q$. Then the distribution of $(Ax_1,\ldots,Ax_q)$ is the same as the distribution of $(Y_0+\sum_{j=1}^r \lambda_{1,j}Y_j,\ldots,Y_0+\sum_{j=1}^r \lambda_{q,j}Y_j)$, where $Y_0,Y_1,\ldots,Y_r$ are i.i.d. random variables taking values in $\F_p^n$ uniformly at random conditioned on $Y_1,\ldots,Y_{\ell}$ being linearly independent. However since if we pick $Y_1,\ldots,Y_r$ independently and uniformly at random, with probability $1-o_{n \rightarrow \infty}(1)$ they will be linearly independent,
by taking $n$ to be sufficiently large this distribution can be made arbitrarily close to the distribution of $(Y_0+\sum_{j=1}^r \lambda_{1,j}Y_j,\ldots,Y_0+\sum_{j=1}^r \lambda_{q,j}Y_j)$, where $Y_0,\ldots,Y_r$ are i.i.d. random variables taking values in $\F_p^n$ uniformly at random.

Thus, we can approximate each term $\Ex \left[\exp(\sum_{i=1}^q \beta(i) \cdot \ff(A x_i))\right]$ by
$
\Ex\left[ \exp(\sum_{i=1}^q \beta(i) \cdot \ff(Y_0+\sum_{j=1}^r \lambda_{i,j}Y_j))\right],
$
which is one of the averages $t^*_{\mathcal{L}_i,\beta_i}(\ff)$. We now conclude the proof, since
$\Ex\left[ \exp(\sum_{i=1}^q \beta(i) \cdot \ff(A X_i))\right]$ where $(X_1,\ldots,X_q) \sim \mu$ can be approximated by an appropriate weighted average of $t^*_{\mathcal{L}_i,\beta_1}(\ff),\ldots,t^*_{\mathcal{L}_{\ell},\beta_{\ell}}(\ff)$.
\end{proof}

\subsection{From field functions to distributional functions}

The next step is to move from functions $\ff:\F_p^n \to \F_p$ to functions whose output lies in some convex set. Once this is accomplished, we can use the same techniques used for studying functions $f:\F_p^n \to \D$  that were used in Section~\ref{sec:D:testing_correlation}.

Let $\Dp$ denote the family of probability measures over $\F_p$. That is, $\Dp \subset \mathbb{R}^p$ is given by
$$
\Dp=\{\mu:\F_p \to [0,1]: \sum_{c \in \F_p} \mu(c) = 1\}.
$$
We identify every element $c \in \F_p$ with its corresponding dirac measure on $\F_p$. That is $c \in \F_p$ is corresponded with the probability measure $\mu_c$ where $\mu_c(c)=1$ and $\mu_c(c')=0$ for all $c' \ne c$. We refer to functions $\Gamma: \F_p^n \to \Dp$ as {\em distributional functions}. Note that they are a superfamily of functions from $\F_p^n$ to $\F_p$, which can be regarded as deterministic functions. Given a distributional function $\Gamma$, we identify it with a distribution over functions from $\F_p^n$ to $\F_p$: If $\FF \sim \Gamma$, the value $\FF(x)$ is independently chosen for every $x \in \F_p^n$ according to the distribution $\Gamma(x)$.

We extend the notion of averages $t^*_{\mathcal{L},\beta}$ to distributional functions $\Gamma$.
For $c \in \F_p$ define the function $\mathfrak{a}_c:\Dp \to \D$ to be
$$
\mathfrak{a}_c(\mu) = \Ex_{z \sim \mu} [\exp(c \cdot z)].
$$
For a distributional function $\Gamma:\F_p^n \to \Dp$ we consider the functions $\mathfrak{a}_c \circ \Gamma:\F_p^n \to \D$, which can equivalently be defined as
$$
(\mathfrak{a}_c \circ \Gamma)(x) = \Ex_{\FF \sim \Gamma} [\exp(c \cdot \FF(x))].
$$
Let $\mathcal{L}=\{L_1,\ldots,L_m\}$ be a system of $m$ linear forms in $k$ variables, and let $\beta \in \F_p^m$. We define
$$
t^*_{\mathcal{L},\beta}(\Gamma) = \Ex_{\X \in (\F_p^n)^k}\left[ \prod_{i=1}^m (\mathfrak{a}_{\beta(i)} \circ \Gamma)(L_i(\X))\right].
$$
Note that for functions $\ff:\F_p^n \to \F_p$ this definition  coincides  with our previous definition.

\begin{claim}
\label{claim:F:avg_F_mu_thesame_t*}
Let $\mathcal{L}=\{L_1,\ldots,L_m\}$ be a system of linear forms in $k$ variables. Let $\beta \in \F_p^m$ be a vector of corresponding coefficients. Then for every distributional function $\Gamma: \F_p^n \to \Dp$, and every $\epsilon>0$, we have
$$
\Pr_{\FF \sim \Gamma}\left[\bigg|t^*_{\mathcal{L},\beta}(\FF)-t^*_{\mathcal{L},\beta}(\Gamma)\bigg| \le \epsilon \right] =1-o_n(1).
$$
\end{claim}

\begin{proof}
The proof follows by a first and second moment estimation, and then applying Chebyshev's inequality.
Let $\FF \sim \Gamma$. Fix $\x \in (\F_p^n)^k$, and consider the random variable
$$
A(\x)=\exp\left(\sum_{i=1}^m \beta(i) \FF(L_i(\x)) \right).
$$
We have $t^*_{\mathcal{L},\beta}(\FF) = \frac{1}{p^{nk}} \sum_{\x \in (\F_p^n)^k} A(\x)$. Note that when $L_1(\x),\ldots,L_m(\x)$ are all distinct, we have
$$
\Ex_{\FF \sim \Gamma}[A(\x)] = \prod_{i=1}^n (\mathfrak{a}_{\beta(i)} \circ \Gamma)(L_i(\x)).
$$
Thus, we get that
$$
\left|\Ex_{\FF \sim \Gamma}[t^*_{\mathcal{L},\beta}(\FF)]-t^*_{\mathcal{L},\beta}(\Gamma)\right| \le \Pr_{\X \in (\F_p^n)^k}[L_1(\X),\ldots,L_m(\X) \textrm{ not all distinct}] \le m^2 p^{-n},
$$
where the second inequality follows by the union bound. We now bound the variance of $t^*_{\mathcal{L},\beta}(\FF)$.
Note that two random variables $A(\x'),A(\x'')$ are independent if $\{L_1(\x'),\ldots,L_m(\x')\}$ and $\{L_1(\x''),\ldots,L_m(\x'')\}$ are disjoint. We thus can bound
$$
\mathrm{Var}_{\FF \sim \Gamma}[t^*_{\mathcal{L},\beta}(\FF)] \le \Pr_{\x',\x'' \in (\F_p^n)^k}[\{L_1(\x'),\ldots,L_m(\x')\} \cap \{L_1(\x''),\ldots,L_m(\x'')\} \ne \emptyset] \le m^2 p^{-n},
$$
where the second inequality follows from the union bound. The claim follows from Chebychev's bound.
\end{proof}

We extend also the notion of correlation to distributional functions. We shorthand $\exp(\Gamma):=\mathfrak{a}_1 \circ \Gamma$, and consider
$$
\|\exp(\Gamma)\|_{u(D_n)} = \sup_{\gg \in D_n} |\ip{\exp(\Gamma),\exp(\gg)}|.
$$
The following claim is the only place in the proof of Theorem~\ref{thm:F:main} that uses the sparsity Condition $\bf A3$.

\begin{claim}
\label{claim:F:avg_F_mu_thesame_exp}
Let $\mathcal{D}=\{D_n\}_{n \in \N}$ be a proper dual family. Then for every distributional function $\Gamma: \F_p^n \to \Dp$ and any $\eps>0$ we have
$$
\Pr_{\FF \sim \Gamma}\left[\bigg|\|\exp(\FF)\|_{u(D_n)} - \|\exp(\Gamma)\|_{u(D_n)}\bigg| \ge \eps \right] =o_n(1).
$$
\end{claim}

\begin{proof}
Fix $\gg \in D_n$, and consider the random variable $\ip{\exp(\FF),\exp(\gg)} = \frac{1}{p^n} \sum_{x \in \F_p^n} \exp(\FF(x)-\gg(x))$. Its expected value is $\ip{\exp(\Gamma),\exp(\gg)}$, and since the values $\{\FF(x): x \in \F_p^n\}$ are chosen independently, we can apply Chernoff's bound and get that
$$
\Pr_{\FF \sim \Gamma}\left[\left|\ip{\exp(\FF),\exp(\gg)} - \ip{\exp(\Gamma),\exp(\gg)}\right| \ge \eps \right] \le 2e^{-c \cdot p^n}
$$
for some constant $c=c(\eps)>0$. By the sparsity Condition $\bf A3$ we get that for every $c'>0$ there exists $n_0$, such that for every $n>n_0$ we have $|D_n| \le p^{c' p^n}$. We conclude the proof by choosing $c'$ small enough such that $p^{c'}<e^c$, and apply the union bound over all $\gg \in D_n$.
\end{proof}

We thus obtain the following lemma, which allows us to consider distributional functions instead of field functions.

\begin{lemma}
\label{lemma:F:extend_correlation_testability}
Suppose that a proper dual family ${\mathcal D}=\{D_n\}$ is correlation testable with $q$ queries. Then for every $\eps>0$, there exists  $\delta \in (0,\eps)$, $n_0 \in \mathbb{N}$, and homogeneous systems of linear forms
$\mathcal{L}_1,\ldots,\mathcal{L}_{\ell}$ with $m_1,\ldots,m_{\ell}$ linear forms, accordingly, and corresponding coefficients $\beta_1 \in \F_p^{m_1},\ldots,\beta_{\ell} \in \F_p^{m_{\ell}}$, such that the closures of the following two sets are disjoint:
$$T_\eps := \{(t^*_{\mathcal{L}_1,\beta_1}(\Gamma),\ldots, t^*_{\mathcal{L}_{\ell},\beta_{\ell}}(\Gamma)) | n >n_0,
\Gamma:\F_p^n \rightarrow \Dp, \|\exp(\Gamma)\|_{u(D_n)} \ge \eps\},$$
and
$$S_\eps := \{(t^*_{\mathcal{L}_1,\beta_1}(\Gamma),\ldots, t^*_{\mathcal{L}_{\ell},\beta_{\ell}}(\Gamma)) | n>n_0, \Gamma:\F_p^n \rightarrow \Dp, \|\exp(\Gamma)\|_{u(D_n)} \le \delta\}.$$
Moreover, the systems $\mathcal{L}_1,\ldots,\mathcal{L}_{\ell}$ have Cauchy-Schwarz complexity at most $q-2$.
\end{lemma}

\begin{proof}
Apply Lemma~\ref{lemma:F:correlation_testable_by_linear_forms_with_few_queries} for $\mathcal{D}$ and $\eps/2$. There exist $\delta \in (0,\eps/4)$ and systems of linear forms $\mathcal{L}_1,\ldots,\mathcal{L}_{\ell}$ with $m_1,\ldots,m_{\ell}$ linear forms, accordingly, along with coefficients $\beta_1 \in \F_p^{m_1},\ldots,\beta_{\ell} \in \F_p^{m_{\ell}}$, such that the closures of the sets
$$T_{\eps/2} := \{(t^*_{\mathcal{L}_1,\beta_1}(\ff),\ldots, t^*_{\mathcal{L}_{\ell},\beta_{\ell}}(\ff)) | n >n_0,
\ff:\F_p^n \rightarrow \F_p, \|\exp(\ff)\|_{u(D_n)} \ge \eps/2\}$$
and
$$S_{\eps/2} := \{(t^*_{\mathcal{L}_1,\beta_1}(\ff),\ldots, t^*_{\mathcal{L}_{\ell},\beta_{\ell}}(\ff)) | n>n_0, \ff:\F_p^n \rightarrow \F_p, \|\exp(\ff)\|_{u(D_n)} \le 2\delta\}$$
are disjoint. For a distributional function $\Gamma:\F_p^n \to \Dp$ define $z(\Gamma):=(t^*_{\mathcal{L}_1,\beta_1}(\Gamma),\ldots, t^*_{\mathcal{L}_{\ell},\beta_{\ell}}(\Gamma))$. Let $c>0$ be the $L_{\infty}$ distance between the closures of $T_{\eps/2}$ and $S_{\eps/2}$, and set $\eps'=c/4$. We will show that for large enough $n>n_0$, if $\Gamma:\F_p^n \to \Dp$ has $\|\exp(\Gamma)\|_{u(D_n)} \ge \eps$ then $z(\Gamma)$ is within $L_{\infty}$ distance $\eps'$ from $T_{\eps/2}$; and if $\|\exp(\Gamma)\|_{u(D_n)} \le \delta$ then $z(\Gamma)$ is within
$L_{\infty}$ distance $\eps'$ from $S_{\eps/2}$; this will conclude the lemma.

Consider first the case where $\|\exp(\Gamma)\|_{u(D_n)} \ge \eps$. Let $\FF \sim \Gamma$. By Claims~\ref{claim:F:avg_F_mu_thesame_t*} and~\ref{claim:F:avg_F_mu_thesame_exp} we have that there exists $n_0$, such that for every $n>n_0$ we have
$$
\Pr_{\FF \sim \Gamma}[\|\exp(\FF)\|_{u(D_n)} \ge \eps/2] \ge 0.99,
$$
and that for every $1 \le i \le \ell$ we have
$$
\Pr_{\FF \sim \Gamma}[|t^*_{\mathcal{L}_i,\beta_i}(\FF)-t^*_{\mathcal{L}_i,\beta_i}(\Gamma)| \le \eps'] \ge 1-\frac{1}{100\ell}.
$$
By the union bound, there exists a specific $\ff:\F_p^n \to \F_p$ such that both conditions hold. That is,
$z(\ff) \in S_{\eps/2}$ and $\|z(\Gamma)-z(\ff)\|_{\infty} \le \eps'$. The case where $\|\exp(\Gamma)\|_{u(D_n)} \le \delta$
is completely analogous.
\end{proof}

We thus study from now on distributional functions $\Gamma:\F_p^n \to \Dp$. The next step is to define averages of such functions with regards to polynomial factors. Let $\mathcal{B}$ be a polynomial factor. We define the average $\E(\Gamma|\mathcal{B}):\F_p^n \to \Dp$ as follows. Assume $\mathcal{B}$ defines a partition $C_1 \dotcup \ldots \dotcup C_b$ of $\F_p^n$. For $x \in C_i$ define $\E(\Gamma|\mathcal{B})(x)$ to be the average of $\Gamma(y)$ over $y \in C_i$,
$$
\E(\Gamma|\mathcal{B})(x) = \frac{1}{|C_i|} \sum_{y \in C_i} \Gamma(y).
$$
Note that $\mathfrak{a}_c \circ \E(\Gamma|\mathcal{B}) \equiv \E(\mathfrak{a}_c \circ \Gamma|\mathcal{B})$.

We say a polynomial factor $\mathcal{B}$ is $(d,\delta)$-good for $\Gamma$ if, informally, $U^{d+1}$ norms cannot
distinguish between $\Gamma$ and $\E(\Gamma|\mathcal{B})$ with  an advantage more than $\delta$. Formally, we say a polynomial factor $\mathcal{B}$ is $(d,\delta)$-good for $\Gamma$ if for all $c \in \F_p$ we have
\begin{equation}
\label{eq:bound_ac_circ_mu}
\|\mathfrak{a}_c \circ \Gamma - \mathfrak{a}_c \circ \E(\Gamma|\mathcal{B})\|_{U^{d+1}} \le \delta.
\end{equation}
We first argue that $(d,\delta)$-good polynomial factors exist for every distributional function.

\begin{claim}
\label{claim:F:good_poly_factors_exist_for_distributional}
Let $\delta>0$, $d<p$ and $r(\cdot)$ be an arbitrary growth function. Then for every distributional function $\Gamma:\F_p^n \to \Dp$ there exists a $(d,\delta)$-good polynomial factor $\mathcal{B}$ with degree $d$, complexity $C \le C_{\max}(p,d,\delta,r(\cdot))$ and rank at least $r(C)$.
\end{claim}

\begin{proof}
Apply Theorem~\ref{thm:decompose_multiple_func} on the set of functions $\{\mathfrak{a}_c \circ \Gamma: c \in \F_p\}$.
\end{proof}

The next claim is similar to Theorem~\ref{thm:avg-approx-funcs}. It shows that if $\Gamma$ is a distributional function, and if $\mathcal{B}$ is a $(d,\delta)$-good polynomial factor for $\Gamma$ where $\delta>0$ is small enough, then averages $t^*_{\mathcal{L},\beta}$ cannot distinguish between $\Gamma$ and $\E(\Gamma|\mathcal{B})$ if the true complexity of $\mathcal{L}$ is at most $d$.

\begin{claim}
\label{claim:F:averages_cant_distinguish_distributional}
Let $\mathcal{L}=\{L_1,\ldots,L_m\}$ be a linear system of true complexity $d$ and Cauchy-Schwarz complexity at most $p$.
For every $\eps>0$, there exists $\delta>0$ such that the following holds. Let $\Gamma:\F_p^n \to \Dp$ be a distributional function and let $\mathcal{B}$ be a $(d,\delta)$-good polynomial factor for $\Gamma$. Then for every $\beta \in \F_p^m$,
$$
|t^*_{\mathcal{L},\beta}(\Gamma) - t^*_{\mathcal{L},\beta}(\E(\Gamma|\mathcal{B}))| \le \eps.
$$
\end{claim}

\begin{proof}
Let $f_i := \mathfrak{a}_{\beta(i)} \circ \Gamma$ for $1 \le i \le m$. We have
$$
t^*_{\mathcal{L},\beta}(\Gamma) = \Ex_{\X \in (\F_p^n)^k} \left[\prod_{i=1}^m f_i(L_i(\X))\right],
$$
and since $\mathfrak{a}_{\beta(i)} \circ \E(\Gamma|\mathcal{B}) \equiv \E(f_i|\mathcal{B})$, we also have
$$
t^*_{\mathcal{L},\beta}(\E(\Gamma|\mathcal{B})) = \Ex_{\X \in (\F_p^n)^k}\left[ \prod_{i=1}^m \E(f_i|\mathcal{B})(L_i(\X))\right].
$$
The claim follows by Theorem~\ref{thm:avg-approx-funcs} and (\ref{eq:bound_ac_circ_mu}).
\end{proof}

We would also require an analog of Proposition~\ref{prop:invariance}. Let $\Gamma:\F_p^n \to \Dp$ be a distributional function, and let $\mathcal{B}$ be a $(d,\delta)$-good polynomial factor for $\Gamma$ defined by polynomials $P_1,\ldots,P_C$. Let $\mathcal{B'}$ be another polynomial factor defined by polynomials $Q_1,\ldots,Q_C$. We define a new hybrid distribution, denoted $\E(\Gamma|\mathcal{B} \rightarrow \mathcal{B'})$ as follows: assume $\E(\Gamma|\mathcal{B})(x)=F(P_1(x),\ldots,P_C(x))$ where $F:\F_p^C \to \Dp$ is some function; we define
$$
\E(\Gamma|\mathcal{B} \to \mathcal{B'})(x):=F(Q_1(x),\ldots,Q_C(x)).
$$

\begin{lemma}
\label{lemma:F:hybrid_distribution_invariance}
Let $\mathcal{L}$ be a homogeneous system of $m$ linear forms. Let $d \ge 1$ be a degree bound, and $\eps>0$ a required error. There exists $r_{\min}=r_{\min}(m,d,\eps)$ such that the following holds. Let $\mathcal{B},\mathcal{B'}$ be polynomial factors of degree at most $d$ defined by $P_1,\ldots,P_C$ and $Q_1,\ldots,Q_C$, respectively. Assume that $\deg(P_i)=\deg(Q_i)$ for all $1 \le i \le C$ and $\rank(\mathcal{B}),\rank(\mathcal{B'}) \ge r_{\min}$. Then for every distributional function $\Gamma:\F_p^n \to \Dp$ and any $\beta \in \F_p^m$ we have
$$
\left| t^*_{\mathcal{L},\beta}(\E(\Gamma|\mathcal{B}))-t^*_{\mathcal{L},\beta}(\E(\Gamma|\mathcal{B}\rightarrow \mathcal{B'})) \right| \le \eps.
$$
\end{lemma}

The proof is identical to the proof of Proposition~\ref{prop:invariance} in~\cite{ComplexityPaper}, and we do not repeat it.

\subsection{Proof of Theorem~\ref{thm:F:main}}

The proof follows very similar lines to the proof of Theorem~\ref{thm:D:main}. We will highlight the changes that need to be made in the proof, and avoid repetition wherever possible.

Let $\mathcal{D}=\{D_n\}_{n \in \N}$  be a proper dual family where $D_n$ is a family of functions from $\F_p^n \to \F_p$. By Lemma~\ref{lemma:F:extend_correlation_testability} we get that for every $\eps>0$ there exists $\delta \in (0,\eps)$, $n_0 \in \N$ and systems of homogeneous linear forms $\mathcal{L}_1,\ldots,\mathcal{L}_{\ell}$ which have Cauchy-Schwarz complexity $\le q-2$, along with coefficients $\beta_1,\ldots,\beta_{\ell}$, such such that the closure of the following two sets are disjoint:
$$T_\eps := \{(t^*_{\mathcal{L}_1,\beta_1}(\Gamma),\ldots, t^*_{\mathcal{L}_{\ell},\beta_{\ell}}(\Gamma)) | n>n_0,
\Gamma:\F_p^n \rightarrow \Dp, \|\exp(\Gamma)\|_{u(D_n)} \ge \eps\},$$
and
$$S_\eps := \{(t^*_{\mathcal{L}_1,\beta_1}(\Gamma),\ldots, t^*_{\mathcal{L}_{\ell},\beta_{\ell}}(\Gamma)) | n>n_0, \Gamma:\F_p^n \rightarrow \Dp, \|\exp(\Gamma)\|_{u(D_n)} \le \delta\}.$$

Let $s \le q-2 < p$ be a bound on the Cauchy-Schwarz complexity of $\mathcal{L}_1,\ldots,\mathcal{L}_{\ell}$
and $d \le s$ be a bound on their true complexity. We define $\mathcal{S}$ in the same way as in the proof of Theorem~\ref{thm:D:main}.

\begin{claim}
\label{claim:F:one_in_S}
Unless all functions in $\mathcal{D}$ are constants, we have $1 \in \mathcal{S}$.
\end{claim}

The proof is identical to the proof of Claim~\ref{claim:D:one_in_S}. The case where $\mathcal{D}$ consists of only constant functions is analyzed in the same way, hence we assume that $1 \in \mathcal{S}$ from now on.

\begin{claim}
\label{claim:F:S_bound}
$\mathcal{S} \subseteq \{1,\ldots,s\}$.
\end{claim}

\begin{proof}
The proof is identical to the proof of Claim~\ref{claim:D:S_bound}. The only difference is the argument why $0^{\ell}$ is in the closure of $S_{\eps}$. Let $\Gamma$ map every element of $\mathbb{F}_p^n$ to the uniform probability distribution over $\F_p$. Then it is easy to verify that $|t_{\mathcal{L}_i,\beta_i}(\Gamma)| = O(p^{-n})$ for all $1 \le i \le \ell$.
\end{proof}

We denote $t:=\max(\mathcal{S})$. Let $(\Gamma_n:\F_p^n \to \Dp)_{n \in \N}$ be a sequence of distributional functions where $\lim_{n \to \infty} \Ex [\exp(\Gamma_n)]=0$. Similar to the proof of Theorem~\ref{thm:D:main}, the proof of Theorem~\ref{thm:F:main} follows from the following two lemmas.

\begin{lemma}
\label{lemma:F:main_rightarrow}
If $\lim_{n \rightarrow \infty} \|\exp(\Gamma_n)\|_{U^{t+1}}=0$ then $\lim_{n \rightarrow \infty} \|\Gamma_n\|_{u(D_n)}=0$. (Should be $\|\exp(\Gamma_n)\|_{u(D_n)}=0$)
\end{lemma}

\begin{lemma}
\label{lemma:F:main_leftarrow}
If $\lim_{n \rightarrow \infty} \|f_n\|_{u(D_n)}=0$ then $\lim_{n \rightarrow \infty} \|f_n\|_{U^{t+1}}=0$.
\end{lemma}

The proof of both lemmas is identical to the proof of Lemmas~\ref{lemma:D:main_rightarrow} and~\ref{lemma:D:main_leftarrow}, where the only difference is that one considers averages $t^*_{\mathcal{L}_i,\beta_i}(\Gamma_n)$ instead of $t_{\mathcal{L}_i,\alpha_i}(f_n)$ and apply the claims
proved in the previous subsection for function $\Gamma:\F_p^n \to \Dp$ instead of their analogs for functions $f:\F_p^n \to \D$. The only lemma whose proof needs to be slightly changed is Lemma~\ref{lem:D:identicalProfile}. We sketch below an analogous version for distributional functions. First, note that for every function $F:\F_p^N \to [0,1]$ there exists a distributional function $\Gamma_F:\F_p^N \to \Dp$ such that $\mathfrak{a}_c \circ \Gamma_F \equiv F$ for all $c \in \F_p \setminus \{0\}$: simply set $\Pr[\Gamma_F(x)=0]=F(x)+(1-F(x))\tfrac{1}{p}$ and $\Pr[\Gamma_F(x)=z]=(1-F(x))\tfrac{1}{p}$ for $z \in \F_p \setminus \{0\}$. Moreover, this implies that for every system of $m$ linear forms $\mathcal{L}$ and coefficients $\beta \in (\F_p \setminus \{0\})^m$,
$$
t^*_{\mathcal{L},\beta}(\Gamma_F) = t_{\mathcal{L}}(F).
$$
The lemma now follows, when in the proof of Lemma~\ref{lem:D:identicalProfile} one replaces $F$ with $\Gamma_F$,  $g_n$ with $Q_n$ (a deterministic distributional function) and $h_n$ with $\Gamma_{h_n}$.

\section{Non-isomorphic connected systems have nonempty interior}
\label{sec:interior}

We prove in this section the following theorem.

\begin{theorem}
\label{thm:InteriorLinearForms}
Let $\mathcal{L}_1, \ldots, \mathcal{L}_k$ be non-isomorphic connected systems of linear forms. For sufficiently large $n \in \mathbb{N}$, the set of points
\begin{equation}
\label{eq:ELS}
\left\{(t_{\mathcal{L}_1}(f),\ldots,t_{\mathcal{L}_k}(f))\ |\  f:\F_p^n \rightarrow [0,1]\right\} \subseteq \mathbb{R}^k
\end{equation}
has a non-empty interior.
\end{theorem}

We divided this section into two parts. In Section~\ref{sec:flaggs} we introduce new notations and develop some preliminary tools. The proof of Theorem~\ref{thm:InteriorLinearForms} is given in Section~\ref{sec:TheProof}.

\subsection{Flagged systems of linear forms\label{sec:flaggs}}

Consider a system of linear forms $\mathcal{L}$.
Recall that $\mathcal{L}$ is connected if there does not exists nonempty $S \subsetneq \mathcal{L}$ such that $\span(S) \cap \span(\mathcal{L} \setminus S) = \{\vec{0}\}$. Suppose that there are subsets $S_1, S_2 \subsetneq \mathcal{L}$ such that
$$\span(S_i) \cap \span(\mathcal{L} \setminus S_i) = \{\vec{0}\},$$
for $i=1,2$. Then for $T=S_1 \cap S_2$, we have
$$\span(T) \cap \span(\mathcal{L} \setminus T) =  \{\vec{0}\}.$$
This in particular shows that (up to the isomorphisms) there is a unique way to partition a system of linear forms $\mathcal{L}$ into disjoint \emph{connected} systems of linear forms $\mathcal{L}_1, \ldots, \mathcal{L}_k$.  We call each one of $\mathcal{L}_1,\ldots,\mathcal{L}_k$ a \emph{connected component} of $\mathcal{L}$. Since connectivity is invariant under isomorphisms we have the following trivial observation.

\begin{observation}
Two systems of linear forms $\mathcal{L}_1$ and $\mathcal{L}_2$ are isomorphic if and only if there is a one to one isomorphic correspondence between their connected components.
\end{observation}

A \emph{$1$-flagged system of linear forms} is a  system of linear forms  $\mathcal{L}$ and a \emph{non-zero} linear form $M \in \span(\mathcal{L})$. We use the notation $\mathcal{L}^M$ to denote such a $1$-flagged system of linear forms. Here  $\mathcal{L}$ is called the \emph{underlying system of linear forms} of $\mathcal{L}^M$. We call $\mathcal{L}_0^{M_0}$ and $\mathcal{L}_1^{M_1}$  \emph{isomorphic}, if there is an invertible linear transformation $T:\span(\mathcal{L}_0) \rightarrow \span(\mathcal{L}_1)$ that maps $M_0$ to $M_1$ and its restriction to $\mathcal{L}_0$ induces an isomorphism between  $\mathcal{L}_0$ and $\mathcal{L}_1$.

Let $\mathcal{L}$ be a system of linear forms in $k$ variables. For a $1$-flagged system of linear forms $\mathcal{L}^M$, and a function $f:\mathbb{F}_p^n \rightarrow \mathbb{C}$ define the function $f^{\mathcal{L}^M}:\mathbb{F}_p^n \rightarrow \mathbb{C}$ by
\begin{equation}
\label{eq:flaggedAvg}
f^{\mathcal{L}^M} : x \mapsto \Ex_{\X \in (\F_p^n)^k} \left[ \prod_{L \in \mathcal{L}} f(L(\X)) \bigg| M(\X) = x\right].
\end{equation}
Note that we have
\begin{equation}
\label{eq:AverageFlag}
t_\mathcal{L}(f) = \Ex_{X \in \F_p^n} \left[ f^{\mathcal{L}^M}(X) \right].
\end{equation}

Let $\mathcal{L}_0^{M_0}$ and $\mathcal{L}_1^{M_1}$ be $1$-flagged systems of linear forms in $\mathbb{F}_p^{k_0}$ and $\mathbb{F}_p^{k_1}$, respectively. We want to define an operation that ``glues'' these two systems to each other by identifying $M_0$ and $M_1$. To this end, first we consider the system of linear forms $\mathcal{L}'$ defined as
$$\mathcal{L}' =\{L \oplus \vec{0} \in \mathbb{F}_p^{k_0+k_1}: L \in \mathcal{L}_0 \} \cup
\{\vec{0} \oplus L \in \mathbb{F}_p^{k_0+k_1}: L \in \mathcal{L}_1 \}.$$
Take any element $M \neq \vec{0}$ in $\mathbb{F}_p^{k_0+k_1-1}$, and any \emph{surjective} linear transformation $T:\mathbb{F}_p^{k_0+k_1} \rightarrow \mathbb{F}_p^{k_0+k_1-1}$ that maps both $\vec{0} \oplus M_0$  and $M_1 \oplus \vec{0}$ to $M$. Then
the \emph{product} of $\mathcal{L}_0^{M_0}$ and $\mathcal{L}_1^{M_1}$ which is denoted by $\mathcal{L}_0^{M_0} \cdot \mathcal{L}_1^{M_1}$ is defined as the $1$-flagged system of linear forms
$$\mathcal{L}_0^{M_0} \cdot \mathcal{L}_1^{M_1} := \left(T(\mathcal{L}')\right)^M.$$
Note that since $T$ is surjective, this definition does not depend (up to isomorphism) on the particular choices of the linear form $M$ and the map $T$.
The restrictions of $T$ to each one of the sets
$$\span\{L \oplus \vec{0} \in \mathbb{F}_p^{k_0+k_1}: L \in \mathcal{L}_0 \},$$
and
$$\span\{\vec{0} \oplus L \in \mathbb{F}_p^{k_0+k_1}: L \in \mathcal{L}_1 \}$$ is invertible, and thus $T$ induces isomorphisms between these sets and their corresponding $\mathcal{L}_i$ ($i=0$ or $1$). Therefore  we shall refer to $\{T(L \oplus \vec{0}) \in \mathbb{F}_p^{k_0+k_1-1}: L \in \mathcal{L}_0 \}^{M}$ and $\{T(\vec{0} \oplus L) \in \mathbb{F}_p^{k_0+k_1-1}: L \in \mathcal{L}_1 \}^{M}$ respectively as copies of $\mathcal{L}_0^{M_0}$ and $\mathcal{L}_1^{M_1}$ in $\mathcal{L}_0^{M_0} \cdot \mathcal{L}_1^{M_1}$. In the sequel, we will frequently identify $1$-flagged system of linear forms with their copies in their product.

%

The definition of the product of $1$-flagged systems of linear forms is motivated by the following fact: Let $\mathcal{L}_0^{M_0}$, $\mathcal{L}_1^{M_1}$, $T$, and $M$ be as above. Let  $f:\mathbb{F}_p^n \rightarrow \mathbb{C}$, $x\in \F_p^n$, and $\X$ be a random variable taking values uniformly in $(\F_p^n)^{k_0+k_1-1}$. Since $T$ is surjective, the two random variables $\prod_{L \in \mathcal{L}_0} f(T(L)(\X))$ and $\prod_{L \in \mathcal{L}_1} f(T(L)(\X))$ are conditionally independent given $M(\X) = x$. Thus it follows from (\ref{eq:flaggedAvg}) that
\begin{equation}
\label{eq:prod_L_prod_flagged}
f^{\mathcal{L}_0^{M_0} \cdot \mathcal{L}_1^{M_1}} = f^{\mathcal{L}_0^{M_0}}  f^{\mathcal{L}_1^{M_1}}.
\end{equation}

\begin{lemma}
\label{lem:connectivity}
Let $\mathcal{L}^M:=\mathcal{L}_0^{M_0} \cdot \mathcal{L}_1^{M_1}$ where $\mathcal{L}_0^{M_0}$ and $\mathcal{L}_1^{M_1}$ are $1$-flagged systems of linear forms such that both $\mathcal{L}_0 \cup \{M_0\}$ and $\mathcal{L}_1 \cup \{M_1\}$ are connected.  Then $\mathcal{L} \cup \{M\}$ is also connected.
\end{lemma}
\begin{proof}
Suppose that $\mathcal{L}_0^{M_0}$ and $\mathcal{L}_1^{M_1}$ are respectively in $k_0$ and $k_1$ variables.
Let $T:\mathbb{F}_p^{k_0+k_1} \rightarrow \mathbb{F}_p^{k_0+k_1-1}$ be as in the definition of the product of two $1$-flagged systems of linear forms given above. Consider a nonempty set $S \subsetneq \mathcal{L} \cup \{M\}$. Suppose to the contrary of the assertion that
\begin{equation}
\label{eq:nonconnected}
\span(S) \cap \span((\mathcal{L} \cup\{M\}) \setminus S) = \{\vec{0}\}.
\end{equation}
We identify $\mathcal{L}_0^{M_0}$ and $\mathcal{L}_1^{M_1}$ with their copies in $\mathcal{L}^M$. In particular, both $M_0$ and $M_1$ are identified with $M$. Since $\mathcal{L}_0 \cup \{M\}$ and $\mathcal{L}_1 \cup \{M\}$ are both connected we have $M \in \span(\mathcal{L}_0 \setminus \{M\})$ and $M \in \span(\mathcal{L}_1 \setminus \{M\})$. Then it follows from (\ref{eq:nonconnected}) that we have $S \neq
\mathcal{L}_i \cup \{M\}$, for $i=1,2$. Also by (\ref{eq:nonconnected}) we have
$$\span(S \cap (\mathcal{L}_0 \cup \{M\})) \cap \span((\mathcal{L}_0 \cup \{M\}) \setminus S) = \{\vec{0}\},$$
and
$$\span(S \cap (\mathcal{L}_1 \cup \{M\})) \cap \span((\mathcal{L}_1 \cup \{M\}) \setminus S) = \{\vec{0}\}.$$
Thus at least one of $\mathcal{L}_0 \cup \{M\}$  or $\mathcal{L}_1 \cup \{M\}$ is not connected which contradicts our assumption.
\end{proof}

For a system of linear forms $\mathcal{L}$ and an $L \in \span(\mathcal{L})$ define $\deg_\mathcal{L}(L)$  to be the number of pairs $(x,y) \in \mathcal{L} \times \mathcal{L}$ satisfying $x+y=L$.

Let $\mathcal{L}^M:=\mathcal{L}_0^{M} \cdot \mathcal{L}_1^{M}$ where $\mathcal{L}_0^{M}$ and $\mathcal{L}_1^{M}$ are $1$-flagged systems of linear forms. It follows from the definition of the product that if $x+y \in \span(\{M\})$ with $x \in \mathcal{L}_0$ and $y \in \mathcal{L}_1$, then both $x,y$ belong to $\span(\{M\})$.
Hence for every $L \in \span(\{M\})$, we have
\begin{equation}
\label{eq:deg1}
\deg_{\mathcal{L}_0}(L) + \deg_{\mathcal{L}_1}(L) \le  \deg_\mathcal{L}(L) \le \deg_{\mathcal{L}_0}(L) + \deg_{\mathcal{L}_1}(L) + |\span(\{M\}) \cap \mathcal{L}_0| + |\span(\{M\}) \cap \mathcal{L}_1|,
\end{equation}
and similarly for $L \in \mathcal{L}_i \setminus \span(\{M\})$ where $i=0,1$, we have
\begin{equation}
\label{eq:deg2}
 \deg_{\mathcal{L}_i}(L) \le \deg_\mathcal{L}(L) \le  \deg_{\mathcal{L}_i}(L)+ 2 |\span(\{M\}) \cap \mathcal{L}_{1-i}|.
\end{equation}

\begin{lemma}
\label{lem:1flaggHighrank}
Let $\mathcal{L}_1^{M_1},\ldots,\mathcal{L}_k^{M_k}$ be non-isomorphic $1$-flagged systems of linear forms such that $\mathcal{L}_i \cup \{M_i\}$ are connected for all $i \in [k]$. For every $N>0$, there exist a $1$-flagged system of linear forms $\mathcal{L}^M$ such that $\rank(\span(\mathcal{L}))>N$, and the underlying systems of linear forms of $\mathcal{L}^M \cdot \mathcal{L}_i^{M_i}$ for $i \in [k]$ are connected and non-isomorphic.
\end{lemma}
\begin{proof}
Let $d>p+N$ be larger than the size of $\mathcal{L}_i$, for every $i \in [k]$. Denote $\e_1=(1,0,\ldots,0) \in \mathbb{F}_p^{d}$, and consider the system of linear forms
$$\mathcal{M}:=\left(\{0\} \times \mathbb{F}_p^{d-1}\right) \cup \left(\{1\} \times \{0,1\}^{d-1}\right)\setminus \{\vec{0},\e_1\} \subseteq \mathbb{F}_p^{d}.$$
We claim that  $\mathcal{M}$ is connected. Indeed $\mathbb{F}_p^{d-1} \setminus \{\vec{0}\} \equiv
\left(\{0\} \times \mathbb{F}_p^{d-1} \right) \setminus \{\vec{0}\} \subsetneq \mathcal{M}$
is trivially connected, and hence if $S \subseteq \mathcal{M}$ is such that $\span(S) \cap \span(\mathcal{M} \setminus S) = \{\vec{0}\},$
then without loss of generality we can assume that  $$\left(\{0\} \times \mathbb{F}_p^{d-1} \right) \setminus \{\vec{0}\} \subseteq S.$$
One can easily  verify that $S$ cannot be equal to $\left(\{0\} \times \mathbb{F}_p^{d-1} \right) \setminus \{\vec{0}\}$. Hence there exists at least one element $L \in S \cap \left(\{1\} \times \{0,1\}^{d-1} \right)$. Then $\mathcal{M} \subseteq \span(\{L\} \cup (\{0\} \times \mathbb{F}_p^{d-1}) ) \subseteq \span(S)$ which together with the assumption $\span(S) \cap \span(\mathcal{M} \setminus S) = \{\vec{0}\}$ shows that $S=\mathcal{M}$. Hence $\mathcal{M}$ is connected.

Note that $2^{d-1} \le \deg_\mathcal{M}(L) \le 4 p^{d-1}$ for every $L \in \mathcal{M}$, and $\deg_{\mathcal{M}}(\e_1) = 2(2^{d-1}-1)$. Furthermore for every $\lambda \in \mathbb{F}_p \setminus \{0,1\}$, we have $\deg_\mathcal{M}(\lambda \e_1) =0$. Also we have $\span(\{e_1\}) \cap \mathcal{M}=\emptyset$.
Set
$$\widetilde{\mathcal{L}}^M := \underbrace{\mathcal{M}^{\e_1} \cdot \ldots \cdot \mathcal{M}^{\e_1}}_{\mbox{$10 p^d$ times}}.$$
and
$$\mathcal{L}^M := (\widetilde{\mathcal{L}} \cup \{M\})^M.$$
By (\ref{eq:deg1}) and (\ref{eq:deg2}), and the above properties of $\mathcal{M}$, we have
$2^{d-1} \le \deg_\mathcal{L}(L) \le 4p^{d-1}$, for every $L \in \mathcal{L} \setminus \{M\}$. Moreover $\deg_\mathcal{L}(M) \ge 10p^{d}$, and
$\deg_\mathcal{L}(L)=0$ for every $L \in \span(\{M\}) \setminus \{\vec{0},M\}$. It also follows from $\span(\{\e_1\}) \cap \mathcal{M}=\emptyset$ that
$\span(\{M\}) \cap \mathcal{L}=\{M\}$.

For every $i \le [k]$, set  $\mathcal{N}_i^{W_i} := \mathcal{L}^M \cdot \mathcal{L}_i^{M_i}$. Then by (\ref{eq:deg1}) and (\ref{eq:deg2}), we have
\begin{itemize}
\item[(i)] $\deg_{\mathcal{N}_i}(W_i) \ge 10 p^{d}$;
\item[(ii)] $2^{d-1} \le \deg_{\mathcal{N}_i}(L) \le 4p^{d-1}+ 2p \le 5p^{d-1}$ for every $L \in \mathcal{N}_i \setminus (\mathcal{L}_i \cup \{W_i\})$;
\item[(iii)] $\deg_{\mathcal{N}_i}(L) \le |\mathcal{L}_i| <2^{d-1}$ for every $L \in \mathcal{L}_i \setminus \{W_i\}$.
\end{itemize}

Since $\mathcal{M}$ is connected and $\e_1 \in \span(\mathcal{M})$, we have that $\mathcal{M} \cup \{\e_1\}$ is also connected. Then Lemma~\ref{lem:connectivity} shows that $\mathcal{L}$ is connected. Now since $\mathcal{L}_i$ are connected, Lemma~\ref{lem:connectivity} implies that $\mathcal{N}_i=\mathcal{N}_i \cup \{W_i\}$ are connected. It remains to show that they are non-isomorphic. But if $\mathcal{N}_i$ is isomorphic to $\mathcal{N}_j$ for some $i,j\in [k]$, then there is a bijection between $\mathcal{N}_i$ and  $\mathcal{N}_j$ that can be extended to an invertible $T:\span(\mathcal{N}_i) \rightarrow \span(\mathcal{N}_j)$. Since such a function, maps $\vec{0}$ to $\vec{0}$, and preserves the degrees, by (i), (ii), and (iii) above, we have $T(W_i)=W_j$ and $\{T(L) : L \in \mathcal{L}_i\} = \mathcal{L}_j$.
Thus the restriction of $T$ to  $\mathcal{L}_i^{M_i}$ is an isomorphism between $\mathcal{L}_i^{M_i}$ and $\mathcal{L}_j^{M_j}$ contradicting our assumption that $\mathcal{L}_i^{M_i}$ and $\mathcal{L}_j^{M_j}$ are non-isomorphic.
\end{proof}

\subsection{Finishing the proof\label{sec:TheProof}.}

We view averages $t_{\mathcal{L}}(f)$ as polynomials in the variables $\{f(x):x \in \F_p^n\}$,
$$
t_{\mathcal{L}}(f) = \frac{1}{p^{nk}} \sum_{\x \in (\F_p^n)^k} \prod_{i=1}^m f(L_i(x)).
$$
We start by proving a technical lemma.

\begin{lemma}
\label{lem:indpendenceOverLowDeg}
Let $\mathcal{L}_1,\ldots,\mathcal{L}_k$ be non-isomorphic connected systems of linear forms. Let $P_1,\ldots,P_k$ be functions mapping every $f:\mathbb{F}_p^n \rightarrow \mathbb{C}$ to $\mathbb{C}$ in the following way. Every $P_i$ is a polynomial of degree at most $d$ in variables $\{f(x): x \in \mathbb{F}_p^n\}$. If $n>d+\max_{i \in [k]} \rank (\span(\mathcal{L}_i))$, and for every $i \in [k]$, $\rank (\span(\mathcal{L}_i)) >d$,  then
$$P_1(f) t_{\mathcal{L}_1}(f)+\ldots+ P_k(f) t_{\mathcal{L}_k}(f) \not \equiv 0,$$
unless $P_i \equiv 0$ for all $i \in [k]$.
\end{lemma}
\begin{proof}
We claim a stronger statement that if at least one of $P_i$ is not divisible by the variable $f(\vec{0})$, then
$$\left. P_1(f) t_{\mathcal{L}_1}(f)+\ldots+ P_k(f) t_{\mathcal{L}_k}(f) \right|_{f(\vec{0})=0}\not \equiv 0.$$
Trivially it suffices to prove this statement for the case where for every $i \in [k]$, no monomial of $P_i$ is divisible by $f(\vec{0})$.
Assume to the contrary that there exist polynomials $P_i$ of degree at most $d$ in variables $\{f(x): x \in \mathbb{F}_p^n\}$ with monomials which are not divisible by $f(\vec{0})$ such that
$$\left. P_1(f) t_{\mathcal{L}_1}(f)+\ldots+ P_k(f) t_{\mathcal{L}_k}(f) \right|_{f(\vec{0})=0} \equiv 0.$$

Without loss of generality assume that for some positive integer $l$, every $\mathcal{L}_i$ is a system of linear forms in $\mathbb{F}_p^l$.
Define the rank of a monomial $\prod_{x \in \mathbb{F}_p^n}f(x)^{\alpha_x}$ to be the rank of $\span(\{x : \alpha_x \neq 0\})$. Let $r_i$ denote the largest rank of a monomial with a non-zero coefficient in $P_i$.

Set $i_0 := \mathrm{argmax}_{i \in [k]} \left(r_i+\rank (\span(\mathcal{L}_i))\right)$. Let non-zero $x_1,\ldots,x_a \in \mathbb{F}_p^n$ be so that $$\rank(\span(\{x_1,\ldots,x_a\}))=r_{i_0},$$ and $\prod_{i=1}^a f(x_i)^{\alpha_i}$ where $\alpha_i>0$ appears with a non-zero coefficient in $P_{i_0}$.
Since  $n \ge d+\rank(\span(\mathcal{L}_{i_0}))$, there exists $\x \in (\F_p^n)^k$ such that $\{L(\x): L \in \mathcal{L}_{i_0}\}$ are all distinct and $\span(\{L(\x): L \in \mathcal{L}_{i_0}\}) \cap \span(\{x_1,\ldots,x_a\}) = \{\vec{0}\}$.
Note that the connectivity of $\mathcal{L}_{i_0}$, and the assumption that $\deg(P_{i_0})<\rank (\span(\mathcal{L}_{i_0}))$ implies that the monomial
%
$$\left(\prod_{i=1}^a x_i^{\alpha_i}\right) \prod_{L \in \mathcal{L}_{i_0}} f(L(\x)),$$
appears with a non-zero coefficient in $P_{i_0}(f) t_{\mathcal{L}_{i_0}}(f)$ (i.e. there is no cancelation).

Suppose that this monomial appears with a non-zero coefficient also for some other $1\le j \le k$ in $P_j(f)\prod_{L \in \mathcal{L}_{j}} f(L(\x'))$, where $\x' \in (\F_p^n)^k$. Then the maximality of $r_{i_0}+\rank (\span(\mathcal{L}_{i_0}))$, connectivity of $\mathcal{L}_j$, and the assumption that $\deg(P_j)<\rank(\span(\mathcal{L}_j))$ shows that $\{L(\x'): L \in \mathcal{L}_j\}=\{L(\x): L \in \mathcal{L}_{i_0}\}$ as multisets. By the assumption that $\{L(\x): L \in \mathcal{L}_{i_0}\}$ are all distinct we get that $\{L(\x'): L \in \mathcal{L}_{j}\}$ are also all distinct. It follows that $\mathcal{L}_j$ is isomorphic to $\mathcal{L}_{i_0}$, which is a contradiction.
\end{proof}

Consider a system of linear forms $\mathcal{L}$ and a function $f:\mathbb{F}_p^n \rightarrow \mathbb{C}$. Define the function $f^{\partial \mathcal{L}}: \mathbb{F}_p^n \rightarrow \mathbb{C}$, as
$$f^{\partial \mathcal{L}}(x) := \sum_{L \in \mathcal{L}} f^{(\mathcal{L}\setminus \{L\})^L}(x).$$
The following easy lemma which follows from linearity of expectation explains the motivation for this notation.
\begin{lemma}
\label{lem:derivative}
For $f,g:\mathbb{F}_p^n \rightarrow \mathbb{C}$, and every system of linear forms $\mathcal{L}$, we have
$$\frac{d}{dt} t_\mathcal{L}(f+tg)|_{t=0} = \Ex \left[ g(X) f^{\partial \mathcal{L}}(X) \right],$$
where $X$ is a random variable taking values in $\mathbb{F}_p^n$ uniformly at random.
\end{lemma}
\begin{proof}
We have 
\begin{eqnarray*}
\frac{d}{dt} t_\mathcal{L}(f+tg)|_{t=0} &=& \frac{d}{dt} \left. \Ex\left[ \prod_{L \in \mathcal{L}} (f+tg)(L(\X)) \right] \right|_{t=0} = \Ex\left[ \sum_{L \in \mathcal{L}} g(L(\X)) \prod_{L \in \mathcal{L} \setminus \{L\}} f(L(\X))\right] \\
&=& \Ex \left[ g(X) f^{\partial \mathcal{L}}(X) \right].
\end{eqnarray*}
\end{proof}

Consider connected non-isomorphic systems of linear forms $\mathcal{L}_1,\ldots,\mathcal{L}_k$. We claim that in order to prove Theorem~\ref{thm:InteriorLinearForms} it suffices to shows that there exists $f:\mathbb{F}_p^n \rightarrow (0,1)$ such that $f^{\partial\mathcal{L}_1},\ldots, f^{\partial\mathcal{L}_k}$ are linearly independent over $\mathbb{R}$.

\begin{claim}
Let $f:\F_p^n \to (0,1)$ be such that $f^{\partial\mathcal{L}_1},\ldots, f^{\partial\mathcal{L}_k}$ are linearly independent over $\mathbb{R}$. Then there exists $\eps>0$ such that
$$
\left\{\left(t_{\mathcal{L}_1}(f),\ldots,t_{\mathcal{L}_k}(f)\right)+z: z \in \mathbb{R}^k, \|z\|_{\infty} \le \eps\right\}
\subseteq
\left\{\left(t_{\mathcal{L}_1}(g),\ldots,t_{\mathcal{L}_k}(g)\right): g:\F_p^n \to (0,1)\right\}.
$$
\end{claim}

\begin{proof}
Let $\e_1,\ldots,\e_k \in \mathbb{R}^k$ denote the unit vectors. Since $f^{\partial\mathcal{L}_1},\ldots, f^{\partial\mathcal{L}_k}$ are linearly independent over $\mathbb{R}$, for every $\e_i$ there exists $g_i:\F_p^n \to \mathbb{R}$ such that
$$
\Ex_{X \in \F_p^n} \left[ g_i(X) f^{\partial \mathcal{L}_j}(X) \right]=\delta_{i,j},
$$
where $\delta_{i,j}=1_{[i=j]}$ is the Kronecker delta function. For $z \in \mathbb{R}^k$ define $g_z(x) = \sum_{i=1}^k z_i g_i(x)$, and consider the map $T:\mathbb{R}^k \to \mathbb{R}^k$ defined as
$$
T(z_1,\ldots,z_k) = (t_{\mathcal{L}_1}(f+g_z),\ldots,t_{\mathcal{L}_k}(f+g_z)).
$$
The Jacobian of $T$ is the identity matrix, and hence invertible. By the inverse function theorem, for every $\eta>0$, $\{T(z): \|z\|_{\infty}<\eta\}$ contains a neighborhood of $T(0)=(t_{\mathcal{L}_1}(f),\ldots,t_{\mathcal{L}_k}(f))$. We will choose $\eta>0$ small enough such that $\|f + g_z\|_{\infty} < 1$ for all $\|z\|_{\infty}<\eta$.
\end{proof}


Suppose to the contrary that for every  $f:\mathbb{F}_p^n \rightarrow (0,1)$, $f^{\partial\mathcal{L}_1},\ldots, f^{\partial\mathcal{L}_k}$ are linearly \emph{dependent} over $\mathbb{R}$. Note that for every $1 \le i \le k$, and for every $x_0 \in \mathbb{F}_p^n$, $f^{\partial\mathcal{L}_i}(x_0)$ is a polynomial of degree $|\mathcal{L}_i|-1$ in the variables $\{f(x) : x \in \mathbb{F}_p^n\}$. The linear dependency of $f^{\partial\mathcal{L}_1},\ldots, f^{\partial\mathcal{L}_k}$ shows that for every $f:\mathbb{F}_p^n \rightarrow (0,1)$, the $k \times k$ matrix whose $ij$-th entry is $\Ex\left[ f^{\partial\mathcal{L}_i}(X) f^{\partial\mathcal{L}_j}(X) \right]$ is singular which in turn implies that the determinant of this matrix as a polynomial in $\{f(x) : x \in \mathbb{F}_p^n\}$ is the zero polynomial. So the functions $f^{\partial\mathcal{L}_i}$ considered as vectors with polynomial entries are dependent over the field of fractions of polynomials in variables $\{f(x) : x \in \mathbb{F}_p^n\}$. Furthermore since the degree of the $ij$-th entry of this matrix is at most $|\mathcal{L}_i|+|\mathcal{L}_j|-2$ which does not depend on $n$, we conclude that there exist polynomials $P_1,\ldots,P_k$ in the variables $\{f(x) : x \in \mathbb{F}_p^n\}$, and of degree at most some integer $C:=C(|\mathcal{L}_1|,\ldots,|\mathcal{L}_k|)$ (which does not depend on $n$) such that
\begin{equation}
\label{eq:dependPols}
P_1(f) f^{\partial\mathcal{L}_1}+\ldots +P_k(f) f^{\partial\mathcal{L}_k} \equiv 0.
\end{equation}
Let $\mathcal{M}_1^{L_{1}},\ldots,\mathcal{M}_l^{L_{l}}$ be some representatives for all the isomorphism classes of $1$-flagged systems of linear forms $\{(\mathcal{L}_i \setminus \{L\})^L: i \in [k], L \in \mathcal{L}_i\}$. Let $\{\alpha_{i,j} \in \mathbb{Z}_+ : i \in [k], j \in [l]\}$ be such that for every $i\in [k]$, we have
$$f^{\partial\mathcal{L}_i} = \sum_{j=1}^l \alpha_{i,j}  f^{\mathcal{M}_j^{L_{j}}}.$$
By Lemma~\ref{lem:1flaggHighrank} it is possible to find a $1$-flagged system of linear form $\mathcal{L}^M$ of arbitrarily large rank such that for $\mathcal{N}_j^{W_j} := \mathcal{M}_{j}^{L_j} \cdot \mathcal{L}^M$ (where $j \in [l]$), the systems of linear forms $\mathcal{N}_j$ are non-isomorphic and connected. Note that by~\eqref{eq:prod_L_prod_flagged} we have $f^{\mathcal{N}_j^{W_{j}}} = f^{\mathcal{M}_j^{L_{j}}} f^{\mathcal{L}^{M}}$, and so we have
$$
\sum_{i=1}^k \sum_{j \in [l]} \alpha_{i,j} P_i(f) f^{\mathcal{N}_j^{W_{j}}} \equiv 0,
$$
which by~\eqref{eq:AverageFlag} implies that
\begin{equation}
\sum_{j \in [l]} \left(\sum_{i=1}^k \alpha_{i,j} P_i(f)\right) t_{\mathcal{N}_j}(f) \equiv 0.
\end{equation}
By Lemma~\ref{lem:indpendenceOverLowDeg} this shows that $\sum_{i=1}^k \alpha_{i,j} P_i(f) \equiv 0$ for every $j \in [l]$.
On the other hand since the $\mathcal{L}_i$ are non-isomorphic, $(\mathcal{L}_i \setminus \{L\})^L \not \equiv (\mathcal{L}_j \setminus \{M\})^M $  for every two distinct $i,j \in [l]$ and every $L \in \mathcal{L}_i$ and $M \in \mathcal{L}_j$. Hence for every $j \in [l]$, there is exactly one $i \in [k]$ such that $\alpha_{i,j} \neq 0$. Also trivially for every $i \in [k]$, there exists at least one $j \in [l]$ with $\alpha_{i,j} \neq 0$. It follows that $P_i \equiv 0$ for every $i \in [k]$ which is a contradiction.

\section{Concluding remarks}
\label{sec:conclusion}

In this paper we study  affine invariant properties which are testable. We show that essentially every such property can be tested by an appropriate Gowers uniformity norm. One technical limitation of our techniques is that they hold only if the field size is not too small (i.e. if the Cauchy-Schwarz complexity is smaller than the field size). The main reason for this obstacle is that the main technical tools developed by the authors in~\cite{ComplexityPaper} are limited to large fields. A recent result of Tao and Ziegler~\cite{TaoZiegler2010} extends the inverse Gowers theorem to small characteristics, and it is plausible that combining the techniques would allow to extend our results to all fields. We leave this for future work.

Even if all the required technical tools were established, it would still not answer the following problem: is it possible to test, using a constant number of queries, whether a function $f:\F_p^n \to \F_p$ is correlated to a polynomial of degree $d$, where $d>p$? We know that the Gowers norm test fails, as it actually tests distance to a larger set of functions (non-classical polynomials). The simplest case which is unknown is the following:

\begin{prob}
Let $f:\F_2^n \to \F_2$. Does there exist a test which queries $f$ on a constant number of positions, and which can distinguish whether $f$ has noticeable or negligible correlation with cubic polynomials?
\end{prob}

\ignore{In fact, even the following simpler problem is unknown:} 
However, recently in~\cite{hatami-lovett-distance} the authors proved a very general theorem that in particular gives a positive answer to the following simpler problem that was stated as an open problem in the previous version of this paper.  

\begin{prob}
Let $\eps>0$ and $f:\F_2^n \to \F_2$. Does there exist a test which queries $f$ on $q(\eps)$ positions, and which can distinguish whether $f$ has correlation at least $\eps$, or at most $\delta(\eps)$, with cubic polynomials?
\end{prob}

\bibliographystyle{plain}

\end{document}